\documentclass[onecolumn,superscriptaddress]{revtex4}
\usepackage[bookmarks = false,colorlinks = true,linkcolor = blue,urlcolor= black,citecolor = blue]{hyperref}
\usepackage{graphicx,bm,bbm}
\usepackage{amsmath,amssymb,amsthm}
\usepackage[misc]{ifsym}
\usepackage{appendix}
\usepackage{enumitem}
\usepackage{subeqnarray}
\newtheorem{theorem}{Theorem}
\newtheorem{lemma}{Lemma}
\newtheorem{corollary}{Corollary}
\newtheorem{definition}{Definition}
\newtheorem{proposition}{Proposition}
\begin{document}
\title{Qubit State Discrimination using Post-measurement Information}
\author{Donghoon Ha}
\affiliation{Department of Applied Mathematics and Institute of Natural Sciences, Kyung Hee University, Yongin 17104, Republic of Korea}
\author{Jeong San Kim}
\email{freddie1@khu.ac.kr}
\affiliation{Department of Applied Mathematics and Institute of Natural Sciences, Kyung Hee University, Yongin 17104, Republic of Korea}
\author{Younghun Kwon}
\email{yyhkwon@hanyang.ac.kr}
\affiliation{Department of Applied Physics, Center for Bionano Intelligence Education and Research, Hanyang University, Ansan 15588, Republic of Korea}
\begin{abstract}
We consider the optimal discrimination of nonorthogonal qubit states with post-measurement information and 
provide an analytic structure of the optimal measurements. We also show that there is always a null optimal measurement when post-measurement information is given.
Further, in discriminating four states using post-measurement information, we analytically provide the optimal probability of correct guessing and show that the uniqueness of optimal measurement is equivalent to the non-existence of non-null optimal measurement with post-measurement information.
\end{abstract}
\maketitle

\section{Introduction}
Whereas orthogonal quantum states can be perfectly discriminated in quantum physics, it is not generally true for nonorthogonal states\cite{ref:chef2000,ref:barn20091,ref:berg2010,ref:bae2015}. For these reasons, various measurement strategies have been studied 
for optimal discrimination of nonorthogonal states, such as \emph{minimum-error discrimination}(ME), unambiguous discrimination, and maximum-confidence discrimination\cite{ref:hels1976,ref:ivan1987,ref:diek1988,ref:pere1988,ref:jaeg1995,ref:ha2015,ref:crok2006}.
ME is a discrimination scheme to minimize the average error probability without inconclusive results.
Although a necessary and sufficient condition for realizing a minimum-error measurement in general cases is well known\cite{ref:hole1979,ref:yuen1975,ref:elda20031,ref:barn20092}, the general solution for ME is not yet known except for ME of two states, symmetric states, and qubit states\cite{ref:hels1976,ref:hunt2004,ref:sams2009,ref:jafa2011,ref:ha2013,ref:ha2014,ref:ban1997,ref:chou2003,ref:bae20131,ref:bae20132}.
In some cases, ME can be performed without the help of measurement, simply by guessing the state with the greatest prior probability is prepared\cite{ref:hunt2003}.\\
\indent When the post-measurement information about the prepared subensemble is available,
some nonorthogonal states can be perfectly discriminated\cite{ref:akib2019}. However, in general, nonorthogonal qubit states cannot be perfectly discriminated even with post-measurement information about the prepared subensemble. Therefore, for the case of qubit state, it is important to investigate \emph{minimizing the average error probability with post-measurement information}(MEPI)\cite{ref:ball2008,ref:gopa2010,ref:carm2018}. Also, it is meaningful since MEPI is known to have a relation with the incompatibility of measurements\cite{ref:carm2018,ref:hein2016,ref:carm2019,ref:skrz20191,ref:uola2019,ref:skrz20192}.\\
\indent MEPI problem can be understood in view of ME problem;
a MEPI of quantum state ensemble consisting of $m$ subensembles with $n_{1},\ldots,n_{m}$ states 
can be translated into a ME of quantum state ensemble with $\prod_{b=1}^{m}n_{b}$ states
by modifying the states and prior probabilities in the original ensemble\cite{ref:gopa2010}.
This approach can be useful for characterizing MEPI of qubit states 
because useful properties and analytical results for ME of qubit states are already well known\cite{ref:hunt2004,ref:ha2013,ref:ha2014,ref:bae20131,ref:bae20132}.\\
\indent In this paper, we analyze MEPI of nonorthogonal qubit states and 
provide an analytic structure of the optimal measurements based on the analysis of some ME problem.
We first show that a null optimal measurement exists for any MEPI of qubit states.
We also analytically provide a necessary and sufficient condition that pre-measurement information is strictly more favorable than post-measurement information when all subensembles have two states.
Moreover, we characterize the optimal measurements by classifying MEPI into the two cases if it is possible or not without the help of measurement. In particular, for the case where the ensemble consists of two subensembles with two states and pre-measurement information is strictly more favorable than post-measurement information, we analytically provide the optimal probability of correct guessing. In this case, we further show that the uniqueness of optimal measurement is equivalent to the non-existence of non-null MEPI measurement.\\
\indent This paper is organized as follows: 
In Sect.~\ref{sec:gamq}, we review and analyze ME of qubit states.
By applying the analysis for ME of qubit states to the ME problem associated with MEPI problem, we provide
our results for MEPI of qubit states in Sect.~\ref{sec:gamf}. 
In Sect.~\ref{sec:conc}, we conclude our results.

\section{Minimum-error Discrimination of Qubit States}\label{sec:gamq}
In two-level quantum systems (qubit), a state is expressed by a density operator on two-dimensional complex Hilbert space $\mathcal{H}$ and
a measurement with a finite outcome set $\Omega=\{1,\ldots,n\}$ is described by a \emph{positive operator valued measure}(POVM)  $\mathcal{M}$,
which is composed of $n$ positive semidefinite operators $M_{1},\ldots,M_{n}$ satisfying $\sum_{i\in\Omega}M_{i}=\mathbbm{1}$.
Here, $\mathbbm{1}$ is the identity operator on $\mathcal{H}$.
We say that $i\in\Omega$ is null(non-null) if $M_{i}\in\Omega$ is zero(non-zero). 
We also say that $\mathcal{M}$ is \emph{null} if it has at least one null outcome, otherwise \emph{non-null}.
\\
\indent In this section, we consider ME of \emph{qubit} state ensemble
$\mathcal{E}=\{\eta_{i},\rho_{i}\}_{i\in\Omega}$ in which the qubit state $\rho_{i}$ is prepared with the probability $\eta_{i}$.
We specify $\eta_{1}$ as the greatest prior probability to reduce the repetitive representation, that is,
\begin{equation}
\eta_{1}\geqslant\eta_{2},\ldots,\eta_{n}.
\end{equation}
A qubit state $\rho_{i}$ can be described using the Bloch vector $\bm{\nu}_{i}$ in the three-dimensional real space $\mathbb{R}^{3}$,
\begin{equation*}
\begin{array}{c}
\rho_{i}=\frac{1}{2}(\mathbbm{1}+\bm{\nu}_{i}\cdot\bm{\sigma}),\ i\in\Omega.
\end{array}
\end{equation*}
where $\bm{\sigma}$ is the Pauli matrices $(\sigma_{X},\sigma_{Y},\sigma_{Z})$. \\
\indent Given a qubit state ensemble $\mathcal{E}=\{\eta_{i},\rho_{i}\}_{i\in\Omega}$
and distinct $i,j\in\Omega$, points in $\mathbb{R}^{3}$, which has $|\eta_{i}-\eta_{j}|$ as distance difference from $\eta_{i}\bm{\nu}_{i}$ and $\eta_{j}\bm{\nu}_{j}$, form a hyperboloid of two sheets.
The one hyperboloid sheet consists of points ${\bm v}\in\mathbb{R}^{3}$
satisfying $\varphi_{i}({\bm v})=\varphi_{j}({\bm v})$, where
\begin{equation}
\varphi_{i}({\bm v})=\eta_{i}+\|\eta_{i}\bm{\nu}_{i}-{\bm v}\|,\ i\in\Omega.
\end{equation}
Here $\|\cdot\|$ is the Euclidean norm. The sheet divides $\mathbb{R}^{3}$ into two sets $\{\bm{v}\in\mathbb{R}^{3}:\varphi_{i}({\bm v})\geqslant\varphi_{j}({\bm v})\}$ and $\{\bm{v}\in\mathbb{R}^{3}:\varphi_{i}({\bm v})<\varphi_{j}({\bm v})\}$.
We use the following definitions to express various conditions in ME of $\mathcal{E}$.

\begin{definition}\label{def:hpm}
For each $\mathsf{S}\subseteq\Omega$, 
\begin{equation}
\begin{array}{rcl}
\mathcal{Z}_{\mathsf{S}}:=\{{\bm v}\in\mathcal{P}_{\mathsf{S}}&:&\varphi_{i}({\bm v})=\varphi_{j}({\bm v})\ \forall (i,j)\in\mathsf{S}\times\mathsf{S},\\
&&\varphi_{i}({\bm v})\geqslant\varphi_{j}({\bm v})\ \forall (i,j)\in\mathsf{S}\times(\Omega-\mathsf{S})\},
\end{array}
\end{equation}
where $\mathcal{P}_{\mathsf{S}}$ is the relative interior of the convex hull of $\{\eta_{i}\bm{\nu}_{i}\}_{i\in\mathsf{S}}$, that is,
\begin{equation}
\begin{array}{c}
\mathcal{P}_{\mathsf{S}}=\{\sum_{i\in\mathsf{S}}c_{i}\eta_{i}\bm{\nu}_{i}\,:\, c_{i}>0\,\forall i\in\mathsf{S},
\  \sum_{i\in\mathsf{S}}c_{i}=1\ \}.
\end{array}
\end{equation}
Also, we denote by $\mathcal{Z}$ the union of all $\mathcal{Z}_{\mathsf{S}}$, that is, $\mathcal{Z}=\bigcup_{\mathsf{S}\subseteq\Omega}\mathcal{Z}_{\mathsf{S}}$.
\end{definition}

\indent In ME of $\mathcal{E}$, we use a POVM $\mathcal{M}=\{M_{i}\}_{i\in\Omega}$ as a measurement such that
the prepared state is guessed to be $\rho_{i}$ for each $i\in\Omega$.  
Then, the maximal average probability of correctly guessing the given qubit state is
\begin{equation}
p_{\rm guess}=\max_{\mathcal{M}}\sum_{i\in\Omega}\eta_{i}{\rm Tr}[\rho_{i}M_{i}].
\end{equation}
ME is the task of finding optimal measurements that provides $p_{\rm guess}$, called the \emph{guessing probability}.
We also note that the guessing probability in ME cannot be less than the greatest prior probability, that is,
\begin{equation}
p_{\rm guess}\geqslant\eta_{1}.
\end{equation}
\indent In some cases, ME can be performed without the help of measurement\cite{ref:hunt2003}; the guessing probability can be obtained by taking the state with the greatest prior probability as the prepared state, that is, $p_{\rm guess}=\eta_{1}$. Even for that case, nontrivial optimal measurements can possibly exist.
The following proposition provides a necessary and sufficient condition for $\mathcal{M}$ to be optimal
when $p_{\rm guess}=\eta_{1}$.
The proof of Proposition~\ref{pro:nomec} is given in Appendix~\ref{app:hnome}.

\begin{proposition}\label{pro:nomec}
For ME of qubit state ensemble $\mathcal{E}$,
\begin{enumerate}[label={\rm(\alph*)}]
\item $p_{\rm guess}=\eta_{1}$ if and only if
$\mathcal{Z}_{\{1\}}$ is not an empty set $\varnothing$ or, equivalently,
\begin{equation}
\epsilon_{i}\geqslant\lambda_{i}\ \forall i\in\Omega,
\end{equation}
where
\begin{equation}
\epsilon_{i}=\eta_{1}-\eta_{i},\ \lambda_{i}=\|\eta_{1}\bm{\nu}_{1}-\eta_{i}\bm{\nu}_{i}\|,\ i\in\Omega.
\end{equation}
\item When $p_{\rm guess}=\eta_{1}$, a POVM $\mathcal{M}$ is optimal if and only if
\begin{equation}\label{eq:trime}
\forall i\neq1,\ 
M_{i}\ \propto\ \left\{
\begin{array}{ccc}
\mathbbm{1}+\frac{\eta_{i}\bm{\nu}_{i}-\eta_{1}\bm{\nu}_{1}}{\|\eta_{i}\bm{\nu}_{i}-\eta_{1}\bm{\nu}_{1}\|}\cdot\bm{\sigma}
&,&\epsilon_{i}=\lambda_{i},\\
0&,&\epsilon_{i}>\lambda_{i}.
\end{array}
\right.
\end{equation}
\end{enumerate}
\end{proposition}

\indent When $p_{\rm guess}=\eta_{1}$, if $\epsilon_{i}>\lambda_{i}$ for all $i\neq1$,
the optimal measurement has only $1\in\Omega$ as a non-null outcome,
and $M_{1}$ becomes the identity operator. That is, the optimal measurement in this case is trivial and unique.
However, if $\epsilon_{i}=\lambda_{i}$ for some $i\neq1$, 
the optimal measurement can have a non-null outcome other than 1,
which implies a non-trivial optimal measurement.
Therefore, when $p_{\rm guess}=\eta_{1}$, a nontrivial optimal measurement exists if and only if
$\epsilon_{i}=\lambda_{i}$ for some $i\neq1$.\\
\indent Now, let us consider the qubit state ensemble $\mathcal{E}$ in which a measurement is a necessary requirement for ME, that is, $p_{\rm guess}>\eta_{1}$.  The following proposition shows that finding $\mathsf{S}\subseteq\Omega$ with $\mathcal{Z}_{\mathsf{S}}\neq\varnothing$ and an element of $\mathcal{Z}$ is directly related to obtaining optimal measurements, where $\varnothing$ is the empty set.
The proof of Proposition~\ref{pro:htp} is given in Appendix~\ref{app:hnome}. 

\begin{proposition}\label{pro:htp}
For ME of qubit state ensemble $\mathcal{E}$, 
\begin{enumerate}[label={\rm(\alph*)}]
\item $\mathcal{Z}$ is always a single-element set $\{\bm{v}\}$ and $p_{\rm guess}=\varphi_{i}(\bm{v})$ for all $i$ in $\mathsf{S}$ with $\mathcal{Z}_{\mathsf{S}}\neq\varnothing$.
\end{enumerate}
When $p_{\rm guess}>\eta_{1}$,
\begin{enumerate}[label={\rm(\alph*)},resume]
\item there is an optimal measurement having $\mathsf{S}\subseteq\Omega$ as the set of all non-null outcomes if and only if $\mathcal{Z}_{\mathsf{S}}\neq\varnothing$.
\item Moreover, a POVM $\mathcal{M}$ having $\mathsf{S}$ as the set of all non-null outcomes is optimal if and only if
\begin{equation}\label{eq:mem}
M_{i}\propto\mathbbm{1}+\frac{\eta_{i}\bm{\nu}_{i}-\bm{v}}{\|\eta_{i}\bm{\nu}_{i}-\bm{v}\|}\cdot\bm{\sigma}
\ \ \forall i\in\mathsf{S}.
\end{equation}
\end{enumerate}
\end{proposition}

\indent When $p_{\rm guess}>\eta_{1}$, all optimal POVMs are characterized by the \emph{single} element of $\mathcal{Z}$.
Proposition~\ref{pro:htp} tells the following  three facts.
First, if $\mathcal{Z}_{\mathsf{S}}$ is empty, 
there is no optimal measurement having $\mathsf{S}$ as the set of all non-null outcomes.
This implies that all optimal measurements are null(that is, $M_{i}=0$ for some $i\in\Omega$) if $\mathcal{Z}_{\Omega}$ is empty.
Second, if $\mathcal{Z}_{\mathsf{S}}$ is non-empty and $\mathcal{Z}\subseteq\mathcal{P}_{\mathsf{S}'}$ for some $\mathsf{S}'\subseteq\mathsf{S}$, then $\mathcal{Z}_{\mathsf{S}'}$ is also non-empty.
This implies that,
if the \emph{affine} dimension\cite{ref:boyd2004,ref:affine} of $\{\eta_{i}\bm{\nu}_{i}\}_{i\in\Omega}$ is $D$,
the guessing probability can be obtained by the detection of $D+1$ qubit states \cite{ref:ha2014}. Third, if $\mathcal{Z}_{\mathsf{S}}$ is non-empty and $\{\eta_{i}\bm{\nu}_{i}\}_{i\in\mathsf{S}}$ forms a simplex with affine dimension $|\mathsf{S}|-1$, the optimal measurement having $\mathsf{S}$ as the set of all non-null outcomes is unique.\\
\indent Condition \eqref{eq:mem} is a necessary and sufficient condition for
a POVM $\mathcal{M}$ having $\mathsf{S}$ as the set of all non-null outcomes to be optimal when $p_{\rm guess}>\eta_{1}$.
Therefore, when $p_{\rm guess}>\eta_{1}$, all optimal measurements are obtained in the following three steps.
The first step is to distinguish whether $\mathcal{Z}_{\mathsf{S}}$ is empty or nonempty for each $\mathsf{S}$ with $|\mathsf{S}|\geqslant2$.
Note that, for $p_{\rm guess}>\eta_{1}$, there exists no subset $\mathsf{S}\subseteq\Omega$ with $|\mathsf{S}|=1$ and $\mathcal{Z}_{\sf S}\neq\varnothing$ because all optimal measurements have more than one non-null outcome.
The second step is to find out what the single element of $\mathcal{Z}$ is.
The final step is to get a POVM $\mathcal{M}$ that satisfies Condition \eqref{eq:mem}. 

\section{Main result: MEPI of Qubit States}\label{sec:gamf}
In this section, we consider MEPI of \emph{qubit} state ensemble, 
\begin{equation}\label{eq:mens}
\mathcal{E}=\bigcup_{b\in\mathsf{B}}\{\eta_{ib},\rho_{ib}\}_{i\in\mathsf{A}_{b}}, 
\end{equation}
where 
\begin{equation}
\begin{array}{lcll}
\mathsf{B}&=&\{1,2,\ldots,m\},& m\geqslant 2,\\
\mathsf{A}_{b}&=&\{1,2,\ldots,n_{b}\},& n_{b}\geqslant 2.
\end{array}
\end{equation}
The ensemble $\mathcal{E}$ consists of $m$ subensembles(we use the term ``subensemble'' regardless of the normalization of prior probability),
\begin{equation}
\mathcal{E}_{b}=\{\eta_{ib},\rho_{ib}\}_{i\in\mathsf{A}_{b}},\ b\in\mathsf{B}.
\end{equation}
Without loss of generality, we assume that
\begin{equation}\label{eq:asmpp}
\eta_{1b}\geqslant\eta_{2b},\ldots,\eta_{n_{b}b}\ \forall b\in\mathsf{B}.
\end{equation}
\indent The classical information $b\in\mathsf{B}$ of the prepared subensemble is provided after a measurement is performed.
We use a POVM $\mathcal{M}=\{M_{\bm{\omega}}\}_{\bm{\omega}\in\Omega}$ to describe a measurement,
where $\Omega$ is the Cartesian product of $\mathsf{A}_{1}$,$\mathsf{A}_{2}$,\ldots,$\mathsf{A}_{m}$, that is,
\begin{equation}
\Omega=\mathsf{A}_{1}\times\mathsf{A}_{2}\times\cdots\times\mathsf{A}_{m}.
\end{equation}
Each outcome $\bm{\omega}=(\omega_{1},\ldots,\omega_{m})\in\Omega$ means that the prepared state is guessed to be $\rho_{\omega_{1}1}$,\,$\rho_{\omega_{2}2}$,\,\ldots, or $\rho_{\omega_{m}m}$ according to post-measurement information $b=1$,\,$2$,\,\ldots, or $m$, respectively.\\
\indent The maximal average probability of correctly guessing the prepared qubit state is 
\begin{equation}\label{eq:pcpg}
p_{\rm guess}^{\rm post}
=\max_{\mathcal{M}}\sum_{b\in\mathsf{B}}\sum_{i\in\mathsf{A}_{b}}\sum_{\substack{\bm{\omega}\in\Omega\\ \omega_{b}=i}}\eta_{ib}{\rm Tr}[\rho_{ib}M_{\bm{\omega}}]
=\max_{\mathcal{M}}\sum_{\bm{\omega}\in\Omega}\tilde{\eta}_{\bm{\omega}}{\rm Tr}[\tilde{\rho}_{\bm{\omega}}M_{\bm{\omega}}],
\end{equation}
where $\tilde{\eta}_{\bm{\omega}}$ and $\tilde{\rho}_{\bm{\omega}}$ are positive numbers and density operators, respectively, such that
\begin{equation}\label{eq:tetrd}
\tilde{\eta}_{\bm{\omega}}=\sum_{b\in\mathsf{B}}\eta_{\omega_{b}b},\ \tilde{\rho}_{\bm{\omega}}=\frac{\sum_{b\in\mathsf{B}}\eta_{\omega_{b}b}\rho_{\omega_{b}b}}{\sum_{b'\in\mathsf{B}}\eta_{\omega_{b'}b'}},\ \bm{\omega}\in\Omega.
\end{equation}
MEPI of $\mathcal{E}$ is the task of finding optimal measurements that provides $p_{\rm guess}^{\rm post}$. From Eq.~\eqref{eq:pcpg}, we can see that $p_{\rm guess}^{\rm post}$ is the guessing probability of qubit state ensemble,
\begin{equation}
\tilde{\mathcal{E}}=\{\tilde{\eta}_{\bm{\omega}},\tilde{\rho}_{\bm{\omega}}\}_{\bm{\omega}\in\Omega}.
\end{equation}
where $\{\tilde{\eta}_{\bm{\omega}}\}_{\bm{\omega}\in\Omega}$ is not normalized, that is, $\sum_{\bm{\omega}\in\Omega}\tilde{\eta}_{\bm{\omega}}>1$.
Therefore, MEPI of $\mathcal{E}$ is equivalent to ME of $\tilde{\mathcal{E}}$;
a POVM $\mathcal{M}$ is optimal for MEPI of $\mathcal{E}$ if and only if  it is optimal for ME of $\tilde{\mathcal{E}}$\cite{ref:gopa2010}. To distinguish between optimal measurements for ME and MEPI,
we use ME and MEPI measurements, respectively.
We also note that the assumption in \eqref{eq:asmpp} implies that $\tilde{\eta}_{\bm{1}}$ is the greatest prior probability of $\tilde{\mathcal{E}}$ and a lower bound of $p_{\rm guess}^{\rm post}$, that is,
\begin{equation}\label{eq:apgte1}
p_{\rm guess}^{\rm post}\geqslant\tilde{\eta}_{\bm{1}}\geqslant \tilde{\eta}_{\bm{\omega}}\,\forall \bm{\omega}\in\Omega,
\end{equation}
where 
\begin{equation}
\bm{1}=(1,1,\ldots,1).
\end{equation}

\subsection{Null MEPI measurement}
\indent Similar to ME, we use Bloch representation of qubit states as
\begin{equation}
\begin{array}{lcll}
\rho_{ib}&=&\frac{1}{2}(\mathbbm{1}+\bm{\nu}_{ib}\cdot\bm{\sigma}),& b\in\mathsf{B},i\in\mathsf{A}_{b},\\[1mm]
\tilde{\rho}_{\bm{\omega}}&=&\frac{1}{2}(\mathbbm{1}+\tilde{\bm{\nu}}_{\bm{\omega}}\cdot\bm{\sigma}),&\bm{\omega}\in\Omega.
\end{array}
\end{equation}
Then, 
\begin{equation}\label{eq:benten}
\sum_{b\in\mathsf{B}}\eta_{\omega_{b}b}\bm{\nu}_{\omega_{b}b}=\tilde{\eta}_{\bm{\omega}}\tilde{\bm{\nu}}_{\bm{\omega}}=:\tilde{\bm{\mu}}_{\bm{\omega}} \quad\forall \bm{\omega}\in\Omega.
\end{equation}
When $m=n_{1}=n_{2}=2$, the affine dimension $D$ of $\{\tilde{\bm{\mu}}_{\bm{\omega}}\}_{\bm{\omega}\in\Omega}$ is less than three because
\begin{equation}\label{eq:mn1n2d2}
\begin{array}{rcl}
\tilde{\bm{\mu}}_{(2,1)}-\tilde{\bm{\mu}}_{(1,1)}=\tilde{\bm{\mu}}_{(2,2)}-\tilde{\bm{\mu}}_{(1,2)}=\eta_{21}\bm{\nu}_{21}-\eta_{11}\bm{\nu}_{11},\\
\tilde{\bm{\mu}}_{(1,2)}-\tilde{\bm{\mu}}_{(1,1)}=\tilde{\bm{\mu}}_{(2,2)}-\tilde{\bm{\mu}}_{(2,1)}=\eta_{22}\bm{\nu}_{22}-\eta_{12}\bm{\nu}_{12}.
\end{array}
\end{equation}
Thus, in this case, ME of $\tilde{\mathcal{E}}$ is possible without detecting every state\cite{ref:ha2014}, and there exists a MEPI measurement of $\mathcal{E}$ that is null(that is, $M_{\bm{\omega}}=0$ for some $\bm{\omega}\in\Omega$).\\
\indent Other than $m=n_{1}=n_{2}=2$, the number of all outcomes, that is, $|\Omega|=\prod_{b\in\mathsf{B}}n_{b}$, is greater than four, and
a null MEPI measurement of $\mathcal{E}$ exists because, for any ME of more than four qubit states, there is a ME measurement that is null\cite{ref:ha2014,ref:davi1978,ref:hunt2004}.

\begin{corollary}\label{cor:null}
For any qubit state ensemble $\mathcal{E}$, a null MEPI measurement of $\mathcal{E}$ exists.
\end{corollary}

\subsection{Upper bound of $p_{\rm guess}^{\rm post}$}
\indent When the classical information $b\in\mathsf{B}$ of the prepared subensemble is known prior to perform a measurement, the maximal average probability of correctly guessing the prepared qubit state, $p_{\rm guess}^{\rm prior}$, can be obtained by performing ME measurement of $\mathcal{E}_{b}$ according to the pre-measurement information $b$. In other words, 
\begin{equation}\label{eq:pgmepr}
p_{\rm guess}^{\rm prior}=\sum_{b\in\mathsf{B}}p_{b}^{\mbox{\rm\tiny ME}},
\end{equation}
where $p_{b}^{\mbox{\rm\tiny ME}}$ is the guessing probability of $\mathcal{E}_{b}$, that is,
\begin{equation}
p_{b}^{\mbox{\rm\tiny ME}}=\max_{\mathcal{M}_{b}}\sum_{i\in{\sf A}_{b}}\eta_{ib}{\rm Tr}[\rho_{ib}M_{ib}],\ b\in\mathsf{B}.
\end{equation}
Here, $\mathcal{M}_{b}$ is a POVM with $n_{b}$ elements $M_{ib}$ indicating the detection of $\rho_{ib}$.\\
\indent Guessing the prepared qubit state using a POVM $\mathcal{M}$ and post-measurement information is equivalent to
guessing the prepared state by performing a POVM $\mathcal{M}_{b}$ consisting of
\begin{equation}\label{eq:mibmo}
M_{ib}=\sum_{\substack{\bm{\omega}\in\Omega\\ \omega_{b}=i}}M_{\bm{\omega}},\ i\in\mathsf{A}_{b},
\end{equation}
according to pre-measurement information $b\in\mathsf{B}$. Given an arbitrary POVM $\mathcal{M}_{b}$ with $n_{b}$ elements $M_{1b},\ldots,M_{n_{b}b}$ for each $b\in\mathsf{B}$, $m$ POVMs $\mathcal{M}_{1},\ldots,\mathcal{M}_{m}$ are called \emph{compatible}
if there is a POVM $\mathcal{M}$ satisfying Eq.~\eqref{eq:mibmo} for all $b\in\mathsf{B}$;
otherwise, they are called \emph{incompatible}\cite{ref:hein2016}. Thus, $p_{\rm guess}^{\rm post}$ is upper bounded as
\begin{equation}\label{eq:pgpgp}
p_{\rm guess}^{\rm post}\leqslant p_{\rm guess}^{\rm prior},
\end{equation}
where the equality holds if and only if there are $m$ ME measurements of $\mathcal{E}_{1},\ldots,\mathcal{E}_{m}$ that are compatible\cite{ref:carm2018}.\\
\indent Obviously, a POVM with the identity operator is compatible with any POVM;
therefore, $p_{\rm guess}^{\rm post}=p_{\rm guess}^{\rm prior}$
if $p_{b}^{\mbox{\rm\tiny ME}}=\eta_{1b}$ for some $b\in\mathsf{B}$.
Moreover, POVMs with the same elements are compatible; thus,
$p_{\rm guess}^{\rm post}=p_{\rm guess}^{\rm prior}$ if there are $m$ optimal POVMs for MEs of $\mathcal{E}_{1},\ldots,\mathcal{E}_{m}$ that have the same elements. The following lemma provides a sufficient condition for $p_{\rm guess}^{\rm post}<p_{\rm guess}^{\rm prior}$.

\begin{lemma}\label{lem:ppsc}
Suppose that, for some $b,b'\in\mathsf{B}$ with $b\neq b'$, the ME measurements for $\mathcal{E}_{b}$ and $\mathcal{E}_{b'}$ are unique and consist of rank-one elements $\{M_{ib}\}_{i\in\mathsf{A}_{b}}$,\,$\{M_{jb'}\}_{j\in\mathsf{A}_{b'}}$
such that $M_{ib}\not\propto M_{jb'}$ for all $i\in\mathsf{A}_{b}$ and all $j\in\mathsf{A}_{b'}$.
Then, $p_{\rm guess}^{\rm post}<p_{\rm guess}^{\rm prior}$.
\end{lemma}
\begin{proof}
Assume that $p_{\rm guess}^{\rm post}=p_{\rm guess}^{\rm prior}$ in which
Eq.~\eqref{eq:mibmo} holds for $b,b'\in\mathsf{B}$.
Then, ${\rm rank}(M_{ib})={\rm rank}(M_{jb'})=1$ implies that $M_{\bm{\omega}}\propto M_{ib}$ for all $\bm{\omega}\in\Omega$ with $\omega_{b}=i$ and $M_{\bm{\omega}}\propto M_{jb'}$ for all $\bm{\omega}\in\Omega$ with $\omega_{b'}=j$. Thus, $M_{ib}\not\propto M_{jb'}$ $\forall i,j$ means $M_{ib}=M_{jb'}=0$ $\forall i,j$
which contradicts 
the POVM completeness.
Therefore, $p_{\rm guess}^{\rm post}<p_{\rm guess}^{\rm prior}$.
\end{proof}

\indent In the case of $n_{b}=2$ and $p_{b}^{\mbox{\rm\tiny ME}}>\eta_{1b}$ for all $b\in\mathsf{B}$, from Helstrom bound\cite{ref:hels1976} or Proposition~\ref{pro:htp}, the optimal POVM elements for ME of $\mathcal{E}_{b}$ are uniquely determined as follows:
\begin{equation}\label{eq:m1bm2b}
M_{1b}=\frac{1}{2}(\mathbbm{1}+\hat{\bm{\mu}}_{2b}\cdot\bm{\sigma}),\
M_{2b}=\frac{1}{2}(\mathbbm{1}-\hat{\bm{\mu}}_{2b}\cdot\bm{\sigma}),\ b\in\mathsf{B},
\end{equation}
where 
\begin{equation}\label{eq:hmuib}
\hat{\bm{\mu}}_{ib}=\frac{\eta_{ib}\bm{\nu}_{ib}-\eta_{1b}\bm{\nu}_{1b}}{\|\eta_{ib}\bm{\nu}_{ib}-\eta_{1b}\bm{\nu}_{1b}\|},\ b\in\mathsf{B}.
\end{equation}
Thus, if $\hat{\bm{\mu}_{2b}}\times\hat{\bm{\mu}}_{2b'}=\bm{0}$ for all $b,b'\in\mathsf{B}$, then $p_{\rm guess}^{\rm post}=p_{\rm guess}^{\rm prior}$ because $m$ optimal POVMs for MEs of $\mathcal{E}_{1},\ldots,\mathcal{E}_{m}$ have the same elements; however, if $\hat{\bm{\mu}}_{2b}\times\hat{\bm{\mu}}_{2b'}\neq\bm{0}$ for some $b,b'\in\mathsf{B}$ with $b\neq b'$, then $p_{\rm guess}^{\rm post}>p_{\rm guess}^{\rm prior}$ from Lemma~\ref{lem:ppsc}.

\begin{corollary}\label{cor:nscpe}
When $n_{b}=2$ for all $b\in\mathsf{B}$, $p_{\rm guess}^{\rm post}<p_{\rm guess}^{\rm prior}$ if and only if $\epsilon_{2b}<\lambda_{2b}$ for all $b\in\mathsf{B}$ and
$\hat{\bm{\mu}}_{2b}\times\hat{\bm{\mu}}_{2b'}\neq\bm{0}$
for some $b,b'\in\mathsf{B}$ with $b\neq b'$.
\end{corollary}

\subsection{MEPI for $p_{\rm guess}^{\rm post}=\tilde{\eta}_{\bm{1}}$}
From Inequalities \eqref{eq:apgte1} and \eqref{eq:pgpgp}, $p_{\rm guess}^{\rm post}$ has the following upper and lower bounds,
\begin{equation}
\tilde{\eta}_{\bm{1}}\leqslant p_{\rm guess}^{\rm post}\leqslant p_{\rm guess}^{\rm prior}. 
\end{equation}
Thus, if $p_{b}^{\mbox{\rm\tiny ME}}=\eta_{1b}$ for all $b\in\mathsf{B}$, then $p_{\rm guess}^{\rm post}=\tilde{\eta}_{\bm{1}}$ in which MEPI of $\mathcal{E}$ is possible without the help of measurement;
the prepared state is guessed to be $\rho_{11}$,\,$\rho_{12}$,\,\ldots, or $\rho_{1m}$ according to post-measurement information $b=1$,\,$2$,\,\ldots, or $m$, respectively.\\
\indent From Proposition~\ref{pro:nomec}, we can see that 
$p_{\rm guess}^{\rm post}=\tilde{\eta}_{\bm{1}}$ if and only if 
\begin{equation}\label{eq:cgpe1}
\tilde{\epsilon}_{\bm{\omega}}\geqslant \tilde{\lambda}_{\bm{\omega}}\ \forall \bm{\omega}\in\Omega, 
\end{equation}
where
\begin{equation}
\tilde{\epsilon}_{\bm{\omega}}=\tilde{\eta}_{\bm{1}}-\tilde{\eta}_{\bm{\omega}},\ \tilde{\lambda}_{\bm{\omega}}=\|\tilde{\bm{\mu}}_{\bm{1}}-\tilde{\bm{\mu}}_{\bm{\omega}}\|,\ \bm{\omega}\in\Omega.
\end{equation}
Also, we can see that $p_{b}^{\mbox{\rm\tiny ME}}=\eta_{1b}$ if and only if 
\begin{equation}\label{eq:cpbe1}
\epsilon_{ib}\geqslant \lambda_{ib}\ \forall i\in\mathsf{A}_{b}, 
\end{equation}
where
\begin{equation}\label{eq:eiblib}
\epsilon_{ib}=\eta_{1b}-\eta_{ib},\ \lambda_{ib}=\|\eta_{1b}\bm{\nu}_{1b}-\eta_{ib}\bm{\nu}_{ib}\|,\ b\in\mathsf{B},i\in\mathsf{A}_{b}.
\end{equation}
The following lemma shows that 
$p_{\rm guess}^{\rm post}=\tilde{\eta}_{\bm{1}}$ is equivalent to $p_{b}^{\mbox{\rm\tiny ME}}=\eta_{1b}$ for all $b\in\mathsf{B}$.

\begin{lemma}\label{lem:pgpte1}
For MEPI of qubit state ensemble $\mathcal{E}$, $p_{\rm guess}^{\rm post}=\tilde{\eta}_{\bm{1}}$ if and only if $\epsilon_{ib}\geqslant \lambda_{ib}$ holds for all $b,i$.
When $p_{\rm guess}^{\rm post}=\tilde{\eta}_{\bm{1}}$, a POVM $\mathcal{M}$ is optimal if and only if
\begin{equation}\label{eq:pemo1}
\forall\bm{\omega}\neq\bm{1},\
M_{\bm{\omega}}\ \propto\ \left\{
\begin{array}{ccl}
\mathbbm{1}+\frac{\tilde{\bm{\mu}}_{\bm{\omega}}-\tilde{\bm{\mu}}_{\bm{1}}}{\|\tilde{\bm{\mu}}_{\bm{\omega}}-\tilde{\bm{\mu}}_{\bm{1}}\|}\cdot\bm{\sigma}
&,&\tilde{\epsilon}_{\bm{\omega}}=\tilde{\lambda}_{\bm{\omega}},\\
0&,&\tilde{\epsilon}_{\bm{\omega}}>\tilde{\lambda}_{\bm{\omega}}.
\end{array}
\right.
\end{equation}
\end{lemma}
\begin{proof}
From the definitions of $\tilde{\eta}_{\bm{\omega}}$ and $\tilde{\bm{\mu}}_{\bm{\omega}}$ in \eqref{eq:tetrd} and \eqref{eq:benten}, we have
\begin{equation}\label{eq:tetlt}
\tilde{\epsilon}_{\bm{\omega}}=\sum_{b\in\mathsf{B}}\epsilon_{\omega_{b}b},\
\tilde{\lambda}_{\bm{\omega}}
\leqslant\sum_{b\in\mathsf{B}}\lambda_{\omega_{b}b}\ \forall\bm{\omega}\in\Omega.
\end{equation}
Thus, if Inequality \eqref{eq:cgpe1} holds for all $\bm{\omega}\in\Omega$ or, equivalently, $p_{\rm guess}^{\rm post}=\tilde{\eta}_{\bm{1}}$, then Inequality \eqref{eq:cpbe1} holds for all $b,i$
because $\tilde{\epsilon}_{\bm{\omega}}=\epsilon_{ib}$ and $\tilde{\lambda}_{\bm{\omega}}=\lambda_{ib}$ for all $\bm{\omega}\in\Omega$ such that $\omega_{b}=i$ for some $b\in\mathsf{B}$ and $\omega_{b'}=1$ for all $b'\neq b$.
Since the converse has already been proved, $p_{\rm guess}^{\rm post}=\tilde{\eta}_{\bm{1}}$ if and only if $\epsilon_{ib}\geqslant\lambda_{ib}$ for all $b,i$. In addition, 
directly from Proposition~\ref{pro:nomec}, we can see that 
Condition \eqref{eq:pemo1} 
is a necessary and sufficient condition for a POVM $\mathcal{M}$ to be optimal when $p_{\rm guess}^{\rm post}=\tilde{\eta}_{\bf 1}$.
\end{proof}

\indent For $p_{\rm guess}^{\rm post}=\tilde{\eta}_{\bm{1}}$ and $\bm{\omega}\in\Omega$, if $\tilde{\epsilon}_{\bm{\omega}}=\tilde{\lambda}_{\bm{\omega}}$, then
\begin{equation}\label{eq:eltll}
\epsilon_{\omega_{b}b}=\lambda_{\omega_{b}b}\ \forall b,\ \tilde{\lambda}_{\bm{\omega}}=\sum_{b\in\mathsf{B}}\lambda_{\omega_{b}b} 
\end{equation}
because
\begin{equation}
\sum_{b\in\mathsf{B}}\lambda_{\omega_{b}b}\leqslant\sum_{b\in\mathsf{B}}\epsilon_{\omega_{b}b}=\tilde{\epsilon}_{\bm{\omega}}=\tilde{\lambda}_{\bm{\omega}}\leqslant\sum_{b\in\mathsf{B}}\lambda_{\omega_{b}b},
\end{equation}
where the first and last inequalities follow from the inequalities in \eqref{eq:cpbe1} and \eqref{eq:tetlt}, respectively. Conversely, if Eq.~\eqref{eq:eltll} holds, $\tilde{\epsilon}_{\bm{\omega}}=\tilde{\lambda}_{\bm{\omega}}$ because
\begin{equation}
\tilde{\epsilon}_{\bm{\omega}}=\sum_{b\in\mathsf{B}}\epsilon_{\omega_{b}b}=
\sum_{b\in\mathsf{B}}\lambda_{\omega_{b}b}=\tilde{\lambda}_{\bm{\omega}}.
\end{equation}
Therefore, $\tilde{\epsilon}_{\bm{\omega}}=\tilde{\lambda}_{\bm{\omega}}$ is equivalent to Eq.~\eqref{eq:eltll} when $p_{\rm guess}^{\rm post}=\tilde{\eta}_{\bm{1}}$.
The last condition of Eq.~\eqref{eq:eltll} is equivalent to that
all unit vectors $\hat{\bm{\mu}}_{\omega_{b}b}$ with $\omega_{b}\neq1$ are the same, that is, for all $b\in\mathsf{B}$ with $\omega_{b}\neq1$,
\begin{equation}
\hat{\bm{\mu}}_{\omega_{b}b}
=\frac{\tilde{\bm{\mu}}_{\bm{\omega}}-\tilde{\bm{\mu}}_{\bm{1}}}{\|\tilde{\bm{\mu}}_{\bm{\omega}}-\tilde{\bm{\mu}}_{\bm{1}}\|},
\end{equation}
where $\hat{\bm{\mu}}_{ib}$ is defined in \eqref{eq:hmuib}.

\begin{corollary}\label{cor:opmnb2}
For MEPI of $\mathcal{E}$ with 
$n_{b}=2$ for all $b\in\mathsf{B}$ and $p_{\rm guess}^{\rm post}=\tilde{\eta}_{\bm{1}}$, 
a POVM $\mathcal{M}$ is optimal if and only if
\begin{equation}
\forall\bm{\omega}\neq\bm{1},\ 
M_{\bm{\omega}}\propto\left\{
\begin{array}{ccl}
\mathbbm{1}+\hat{\bm{\mu}}\cdot\bm{\sigma}&,&\epsilon_{2b}=\lambda_{2b}\,\forall\omega_{b}\neq1\ \mbox{\rm and}\\[1mm]
&&\exists\hat{\bm{\mu}}\in\mathbb{R}^{3}\ \mbox{\rm such that}\ \hat{\bm{\mu}}=\hat{\bm{\mu}}_{2b} \,\forall \omega_{b}\neq1, \\[1mm]
0&,&\mbox{\rm otherwise.}
\end{array}
\right.
\end{equation}
\end{corollary}

\subsection{MEPI for $p_{\rm guess}^{\rm post}>\tilde{\eta}_{\bm{1}}$}
\indent In order to consider the case of $p_{\rm guess}^{\rm post}>\tilde{\eta}_{\bm{1}}$,  
we redefine $\mathcal{Z}_{\mathsf{S}}$ of Definition~\ref{def:hpm} suitable for ME of $\tilde{\mathcal{E}}$.

\begin{definition}\label{def:zss}
For each $\mathsf{S}\subseteq\Omega$, 
\begin{equation}\label{eq:dzss}
\begin{array}{rcl}
\mathcal{Z}_{\mathsf{S}}:=\{\bm{v}\in\mathcal{P}_{\mathsf{S}}&:&\varphi_{\bm{\omega}}(\bm{v})=\varphi_{\bm{\omega}'}(\bm{v})\ \forall(\bm{\omega},\bm{\omega}')\in\mathsf{S}\times\mathsf{S},\\
&&\varphi_{\bm{\omega}}(\bm{v})\geqslant\varphi_{\bm{\omega}'}(\bm{v})\ \forall(\bm{\omega},\bm{\omega}')\in\mathsf{S}\times(\Omega-\mathsf{S})\},
\end{array}
\end{equation}
where
\begin{equation}\label{eq:defvps}
\begin{array}{ccl}
\varphi_{\bm{\omega}}({\bm v})&=&\tilde{\eta}_{\bm{\omega}}+\|\tilde{\bm{\mu}}_{\bm{\omega}}-{\bm v}\|,\ \bm{\omega}\in\Omega,\\
&&\\
\mathcal{P}_{\mathsf{S}}&=&\{\sum_{\bm{\omega}\in\mathsf{S}}c_{\bm{\omega}}\tilde{\bm{\mu}}_{\bm{\omega}}\,:\,
c_{\bm{\omega}}>0\ \forall \bm{\omega}\in\mathsf{S},\ \sum_{\bm{\omega}\in\mathsf{S}}c_{\bm{\omega}}=1\}.
\end{array}
\end{equation}
\noindent
We also use $\mathcal{Z}$ to denote the union of all $\mathcal{Z}_{\mathsf{S}}$, that is, $\mathcal{Z}=\bigcup_{\mathsf{S}\subseteq\Omega}\mathcal{Z}_{\mathsf{S}}$.
\end{definition}

\indent From Propositions~\ref{pro:htp}, we can obtain the following lemma, showing that all MEPI measurements of $\mathcal{E}$ are characterized by the \emph{single} element of $\mathcal{Z}$ when $p_{\rm guess}^{\rm post}>\tilde{\eta}_{\bm{1}}$.

\begin{lemma}\label{lem:optmea}
For MEPI of qubit state ensemble $\mathcal{E}$,
$\mathcal{Z}$ is always a single-element set $\{\bm{v}\}$.
If $\mathcal{Z}_{\mathsf{S}}$ is nonempty, 
\begin{equation}
p_{\rm guess}^{\rm post}=\varphi_{\bm{\omega}}(\bm{v})\ \ \forall \bm{\omega}\in\mathsf{S}.
\end{equation}
When $p_{\rm guess}^{\rm post}>\tilde{\eta}_{\bm{1}}$,
there is an optimal measurement having $\mathsf{S}\subseteq\Omega$ as the set of all non-null outcomes if and only if $\mathcal{Z}_{\mathsf{S}}$ is nonempty.
Moreover, a POVM $\mathcal{M}$ having $\mathsf{S}$ as the set of all non-null outcomes is optimal if and only if 
\begin{equation}\label{eq:mepic}
M_{\bm{\omega}}\propto
\mathbbm{1}+\frac{\tilde{\bm{\mu}}_{\bm{\omega}}-{\bm v}}{\|\tilde{\bm{\mu}}_{\bm{\omega}}-{\bm v}\|}\cdot\bm{\sigma}
\ \ \forall\bm{\omega}\in\mathsf{S}.
\end{equation}
\end{lemma}

\begin{figure*}[!tt]
\centerline{
\includegraphics*[bb=0 0 1680 450,scale=0.3]{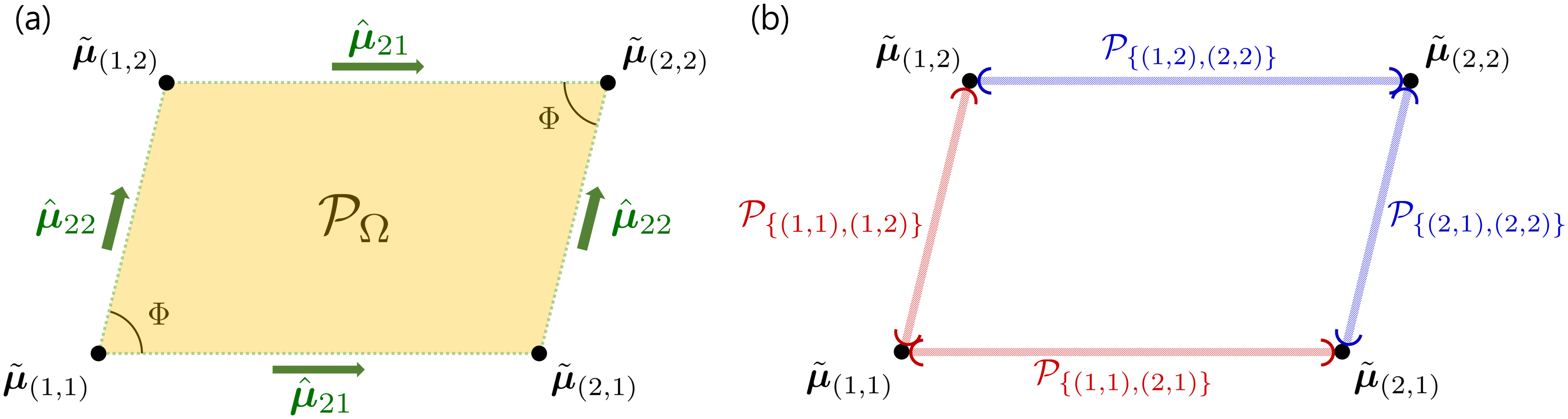}
}
\caption{In the case of $m=n_{1}=n_{2}=2$ and $p_{\rm guess}^{\rm post}<p_{\rm guess}^{\rm prior}$, four vectors $\{\tilde{\bm{\mu}}_{\bm{\omega}}\}_{\bm{\omega}\in\Omega}$ in $\mathbb{R}^{3}$ forms a parallelogram with nonempty interior $\mathcal{P}_{\Omega}$(yellow in (a)).
For $\mathcal{Z}\not\subseteq\mathcal{P}_{\Omega}$, the single element of $\mathcal{Z}$ is in one of two edges $\mathcal{P}_{\{(1,1),(1,2)\}}$ and $\mathcal{P}_{\{(1,1),(2,1)\}}$(red in (b)), but it cannot be in one of two edges $\mathcal{P}_{\{(1,2),(2,2)\}}$ and $\mathcal{P}_{\{(2,1),(2,2)\}}$(blue in (b)).
}\label{fig:para1}
\end{figure*}
\indent For $p_{\rm guess}^{\rm post}>\tilde{\eta}_{\bf 1}$, there is no subset $\mathsf{S}\subseteq\Omega$ with $|\mathsf{S}|=1$ and $\mathcal{Z}_{\sf S}\neq\varnothing$ because all MEPI measurements have more than one non-null outcome. Therefore, $\mathcal{Z}\not\subseteq\mathcal{P}_{\sf S}$ for all ${\sf S}\subseteq\Omega$ with $|\mathsf{S}|=1$.\\
\indent For example, let us consider the case of $m=n_{1}=n_{2}=2$ and $p_{\rm guess}^{\rm post}<p_{\rm guess}^{\rm prior}$. From Corollary~\ref{cor:nscpe}, Lemma~\ref{lem:pgpte1}, and Eq.~\eqref{eq:mn1n2d2}, $p_{\rm guess}^{\rm post}>\tilde{\eta}_{\bf 1}$ and the polygon formed by $\{\tilde{\bm{\mu}}_{\bm{\omega}}\}_{\bm{\omega}\in\Omega}$ is a \emph{parallelogram} with nonempty interior $\mathcal{P}_{\Omega}$ illustrated in Fig.~\ref{fig:para1}.
The interior angle between two unit vector $\hat{\bm{\mu}}_{21}$ and $\hat{\bm{\mu}}_{22}$ defined in \eqref{eq:hmuib} is $\Phi\in(0,\pi)$ satisfying
\begin{equation}
\hat{\bm{\mu}}_{21}\cdot\hat{\bm{\mu}}_{22}=\cos\Phi.
\end{equation}
Thus, we can classify this MEPI into two cases if the single element of $\mathcal{Z}$ is in $\mathcal{P}_{\Omega}$ or not.
\\
\indent Before analyzing two cases $\mathcal{Z}\not\subseteq\mathcal{P}_{\Omega}$ and $\mathcal{Z}\subseteq\mathcal{P}_{\Omega}$, respectively, let us consider the uniqueness of MEPI measurement. When $m=n_{1}=n_{2}=2$ and $p_{\rm guess}^{\rm post}<p_{\rm guess}^{\rm prior}$, Lemma~\ref{lem:optmea} implies that,
for each $\mathsf{S}\subsetneq\Omega$ with $|\mathsf{S}|\geqslant2$,
the MEPI measurement having $\mathsf{S}$ as the set of all non-null outcomes is unique
because $\{\tilde{\bm{\mu}}_{\bm{\omega}}\}_{\bm{\omega}\in\mathsf{S}}$ forms a simplex with affine dimension $|\mathsf{S}|-1$. 
Since all convex combination of two different MEPI measurements are also MEPI measurements, if a non-null MEPI measurement does not exist(that is, $\mathcal{Z}_{\Omega}=\varnothing$), there cannot be more than one
$\mathsf{S}\subseteq\Omega$ satisfying $\mathcal{Z}_{\mathsf{S}}\neq\varnothing$; thus, the MEPI measurement is unique.
However, if a non-null MEPI measurements exists(that is, $\mathcal{Z}_{\Omega}\neq\varnothing$), the MEPI measurement is not unique because a null MEPI measurement also exists from Corollary~\ref{cor:null}. 

\begin{corollary}\label{cor:muni}
When $m=n_{1}=n_{2}=2$ and $p_{\rm guess}^{\rm post}<p_{\rm guess}^{\rm prior}$, 
the MEPI measurement is unique if and only if a non-null MEPI measurement does not exist or, equivalently, $\mathcal{Z}_{\Omega}=\varnothing$.
\end{corollary}

\indent Let us now analyze MEPI of two cases, $\mathcal{Z}\not\subseteq\mathcal{P}_{\Omega}$ and $\mathcal{Z}\subseteq\mathcal{P}_{\Omega}$.
If $\mathcal{Z}\not\subseteq\mathcal{P}_{\Omega}$, then
$\mathcal{Z}\subseteq\mathcal{P}_{\mathsf{S}}$ for only one of the following four $\mathsf{S}$'s:
\begin{equation}\label{eq:sbpg}
\{(1,1),(1,2)\},\,\{(1,1),(2,1)\},\,\{(1,2),(2,2)\},\,\{(2,1),(2,2)\}.
\end{equation}
In other words, the single element of $\mathcal{Z}$ is in one of the edges illustrated in Fig.~\ref{fig:para1}.
Thus, $\mathcal{Z}_{\mathsf{S}}$ is nonempty for only one $\mathsf{S}$ in Eq.~\eqref{eq:sbpg}.
The following theorem provides a necessary and sufficient condition that $\mathcal{Z}_{\mathsf{S}}$ is nonempty for one $\mathsf{S}$ in Eq.~\eqref{eq:sbpg} when $m=n_{1}=n_{2}=2$ and $p_{\rm guess}^{\rm post}<p_{\rm guess}^{\rm prior}$.
The proof of Theorem~\ref{thm:twom} is given in Appendix~\ref{app:thmp}. 

\begin{theorem}\label{thm:twom}
For MEPI of $\mathcal{E}$ with
$m=n_{1}=n_{2}=2$ and $p_{\rm guess}^{\rm post}<p_{\rm guess}^{\rm prior}$, the followings are true:
\begin{equation}\label{eq:notsq}
\begin{array}{l}
\mathcal{Z}_{\{(1,1),(1,2)\}}\neq\varnothing \ \Leftrightarrow\ 
\frac{\epsilon_{21}+\lambda_{21}\cos\Phi}{\lambda_{21}+\epsilon_{21}}\geqslant\frac{\lambda_{21}-\epsilon_{21}}{\lambda_{22}-\epsilon_{22}},
\frac{\epsilon_{21}-\lambda_{21}\cos\Phi}{\lambda_{21}-\epsilon_{21}}\geqslant\frac{\lambda_{21}+\epsilon_{21}}{\lambda_{22}+\epsilon_{22}},\\[1.5mm]
\mathcal{Z}_{\{(1,1),(2,1)\}}\neq\varnothing \ \Leftrightarrow\ 
\frac{\epsilon_{22}+\lambda_{22}\cos\Phi}{\lambda_{22}+\epsilon_{22}}\geqslant\frac{\lambda_{22}-\epsilon_{22}}{\lambda_{21}-\epsilon_{21}}, 
\frac{\epsilon_{22}-\lambda_{22}\cos\Phi}{\lambda_{22}-\epsilon_{22}}\geqslant\frac{\lambda_{22}+\epsilon_{22}}{\lambda_{21}+\epsilon_{21}},\\[1.5mm]
\mathcal{Z}_{\{(2,1),(2,2)\}}=\mathcal{Z}_{\{(1,2),(2,2)\}}=\varnothing,
\end{array}
\end{equation}
where $\epsilon_{ib}$ and $\lambda_{ib}$ are defined in \eqref{eq:eiblib}, and
\begin{equation}\label{eq:pc0110}
p_{\rm guess}^{\rm post}=\left\{
\begin{array}{clc}
\eta_{11}+\frac{1}{2}(\eta_{12}+\eta_{22}+\lambda_{22})&,&\mathcal{Z}_{\{(1,1),(1,2)\}}\neq\varnothing,\\[1mm]
\eta_{12}+\frac{1}{2}(\eta_{11}+\eta_{21}+\lambda_{21})&,&\mathcal{Z}_{\{(1,1),(2,1)\}}\neq\varnothing.\\
\end{array}
\right.
\end{equation}
\end{theorem}

\indent When $\mathcal{Z}_{\{(1,1),(1,2)\}}$ is nonempty,  
$p_{\rm guess}^{\rm prior}-p_{\rm guess}^{\rm post}=\frac{1}{2}(\lambda_{21}-\epsilon_{21})$
which is a finite gap between the guessing probability and the greatest prior probability of $\mathcal{E}_{1}$.
In this case, $p_{\rm guess}^{\rm post}$ can be obtained by ME of $\mathcal{E}_{2}$.
More explicitly, we perform the ME measurement of $\mathcal{E}_{2}$ before obtaining the classical information $b$ about the prepared subensemble.
If $b=1$, we ignore the result and guess the prepared state as $\rho_{11}$, whereas,
if $b=2$, we accept the measurement result as it is.
Then, the average probability of correctly guessing the prepared state becomes $p_{\rm guess}^{\rm post}$.
Similarly, MEPI of $\mathcal{E}$ with $\mathcal{Z}_{\{(1,1),(2,1)\}}\neq\varnothing$ can be achieved with ME of $\mathcal{E}_{1}$.\\
\begin{figure*}[!tt]
\centerline{
\includegraphics*[bb=0 0 1630 440,scale=0.3]{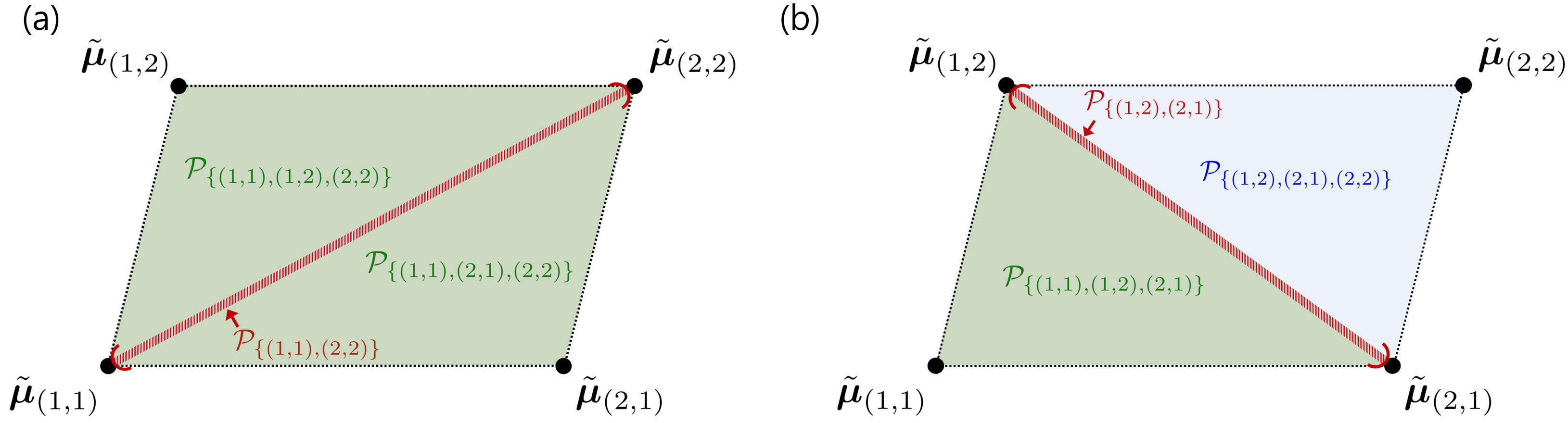}
}
\caption{When $m=n_{1}=n_{2}=2$ and $p_{\rm guess}^{\rm post}<p_{\rm guess}^{\rm prior}$, if $\mathcal{Z}\subseteq\mathcal{P}_{\Omega}$, the single element of $\mathcal{Z}$ exists in one of
the interiors of two line segments, $\mathcal{P}_{\{(1,1),(2,2)\}}$(red in (a)) and $\mathcal{P}_{\{(1,2),(2,1)\}}$(red in (b)),
or in one of the interiors of four triangles, $\mathcal{P}_{\{(1,1),(1,2),(2,2)\}}$, $\mathcal{P}_{\{(1,2),(2,1),(2,2)\}}$(green in (a)), $\mathcal{P}_{\{(1,1),(1,2),(2,1)\}}$(green in (b)), and $\mathcal{P}_{\{(1,2),(2,1),(2,2)\}}$(blue in (b)), but it cannot be an element of $\mathcal{Z}_{\{(1,2),(2,1),(2,2)\}}$.
}\label{fig:para2}
\end{figure*}
\indent If $\mathcal{Z}\subseteq\mathcal{P}_{\Omega}$ equivalent to
\begin{equation}
\mathcal{Z}_{\{(1,1),(1,2)\}}=\mathcal{Z}_{\{(1,1),(2,1)\}}=\varnothing
\end{equation}
from Theorem~\ref{thm:twom}, then $\mathcal{Z}\subseteq\mathcal{P}_{\mathsf{S}}$ for some $\mathsf{S}$ in the following seven $\mathsf{S}$'s:
\begin{equation}\label{eq:zinp}
\begin{array}{r}
\Omega,\{(1,1),(2,2)\},\,\{(1,1),(1,2),(2,2)\},\,\{(1,1),(2,1),(2,2)\},\\
\{(1,2),(2,1)\},\,\{(1,1),(1,2),(2,1)\},\,\{(1,2),(2,1),(2,2)\}.
\end{array}
\end{equation}
In other words, the single element of $\mathcal{Z}$ exists in one of the interiors of line segments or triangles illustrated in Fig.~\ref{fig:para2}.
Thus, $\mathcal{Z}_{\mathsf{S}}$ is nonempty for some $\mathsf{S}$ in \eqref{eq:zinp}.
The following theorem provides a necessary and sufficient condition that $\mathcal{Z}_{\mathsf{S}}\neq\varnothing$ for one $\mathsf{S}$ in \eqref{eq:zinp} when $m=n_{1}=n_{2}=2$, $p_{\rm guess}^{\rm post}<p_{\rm guess}^{\rm prior}$, and $\mathcal{Z}\subseteq\mathcal{P}_{\Omega}$.
The proof of Theorem~\ref{thm:prfthr} is given in Appendix~\ref{app:thmp}. 

\begin{theorem}\label{thm:prfthr}
For MEPI of $\mathcal{E}$ with
$m=n_{1}=n_{2}=2$ and $p_{\rm guess}^{\rm post}<p_{\rm guess}^{\rm prior}$, the followings are true.
\begin{enumerate}[label={\rm(\alph*)}]
\item We have
\begin{equation}\label{eq:nscthrpm}
\begin{array}{lclll}
\mathcal{Z}_{\{(1,1),(2,2)\}}\neq\varnothing&\ \Leftrightarrow\ &\alpha\geqslant|\beta_{+}|/\gamma_{+},\\[1.3mm]
\mathcal{Z}_{\{(1,2),(2,1)\}}\neq\varnothing&\ \Leftrightarrow\ &\alpha\leqslant-|\beta_{-}|/\gamma_{-}, 
\end{array}
\end{equation}
and
\begin{equation}\label{eq:mspth}
p_{\rm guess}^{\rm post}=\left\{
\begin{array}{cclc}
\frac{1}{2}(1+\gamma_{+})&,&\mathcal{Z}_{\{(1,1),(2,2)\}}\neq\varnothing,\\[1mm]
\frac{1}{2}(1+\gamma_{-})&,&\mathcal{Z}_{\{(1,2),(2,1)\}}\neq\varnothing,\\[1mm]
\end{array}
\right.
\end{equation}
where
\begin{eqnarray}\label{eq:glpmrpm}
\alpha&=&\lambda_{21}\lambda_{22}\cos\Phi-\epsilon_{21}\epsilon_{22},\nonumber\\
\beta_{\pm}&=&\epsilon_{21}\lambda_{22}^{2}\mp \epsilon_{22}\lambda_{21}^{2}\pm (\epsilon_{21}\mp \epsilon_{22})\lambda_{21}\lambda_{22}\cos\Phi,\nonumber\\
\gamma_{\pm}&=&\sqrt{\lambda_{21}^{2}+\lambda_{22}^{2}\pm 2\lambda_{21}\lambda_{22}\cos\Phi}.
\end{eqnarray}
Here, $\epsilon_{ib}$ and $\lambda_{ib}$ are defined in \eqref{eq:eiblib}.
\item We have
\begin{equation}\label{eq:nscthrpm1}
\begin{array}{lclll}
\mathcal{Z}_{\{(1,1),(1,2),(2,1)\}}\neq\varnothing&\ \Leftrightarrow\ &
\left\{
\begin{array}{l}
\frac{\epsilon_{21}+\lambda_{21}\cos\Phi}{\lambda_{21}+\epsilon_{21}}<\frac{\lambda_{21}-\epsilon_{21}}{\lambda_{22}-\epsilon_{22}},\\[1.5mm]
\frac{\epsilon_{22}+\lambda_{22}\cos\Phi}{\lambda_{22}+\epsilon_{22}}<\frac{\lambda_{22}-\epsilon_{22}}{\lambda_{21}-\epsilon_{21}},\\[1.5mm]
0\geqslant\alpha>-\beta_{-}/\gamma_{-},
\end{array}
\right.\\[8mm]
\mathcal{Z}_{\{(1,1),(1,2),(2,2)\}}\neq\varnothing&\ \Leftrightarrow\ &
\left\{
\begin{array}{l}
\frac{\epsilon_{21}-\lambda_{21}\cos\Phi}{\lambda_{21}-\epsilon_{21}}<\frac{\lambda_{21}+\epsilon_{21}}{\lambda_{22}+\epsilon_{22}},\\[1.5mm]
\frac{\epsilon_{22}+\lambda_{22}\cos\Phi}{\lambda_{22}+\epsilon_{22}}>-\frac{\lambda_{22}-\epsilon_{22}}{\lambda_{21}-\epsilon_{21}},\\[1.5mm]
0\leqslant\alpha<-\beta_{+}/\gamma_{+},
\end{array}
\right.\\[8mm]
\mathcal{Z}_{\{(1,1),(2,1),(2,2)\}}\neq\varnothing&\ \Leftrightarrow\ &
\left\{
\begin{array}{l}
\frac{\epsilon_{22}-\lambda_{22}\cos\Phi}{\lambda_{22}-\epsilon_{22}}<\frac{\lambda_{22}+\epsilon_{22}}{\lambda_{21}+\epsilon_{21}},\\[1.5mm]
\frac{\epsilon_{21}+\lambda_{21}\cos\Phi}{\lambda_{21}+\epsilon_{21}}>-\frac{\lambda_{21}-\epsilon_{21}}{\lambda_{22}-\epsilon_{22}},\\[1.5mm]
0\leqslant\alpha<\beta_{+}/\gamma_{+},
\end{array}
\right.\\[8mm]
\mathcal{Z}_{\{(1,2),(2,1),(2,2)\}}=\varnothing,&&
\end{array}
\end{equation}
and
\begin{equation}\label{eq:mspth1}
p_{\rm guess}^{\rm post}=\left\{
\begin{array}{cclc}
\tilde{\eta}_{(1,1)}+\frac{\lambda_{21}^{2}-\epsilon_{21}^{2}}{2[\lambda_{21}\cos(\Theta_{-}-\Xi_{-})+\epsilon_{21}]}&,&\mathcal{Z}_{\{(1,1),(1,2),(2,1)\}}\neq\varnothing,\\[1mm]
\tilde{\eta}_{(1,2)}-\frac{\lambda_{21}^{2}-\epsilon_{21}^{2}}{2[\lambda_{21}\cos(\Theta_{+}+\Xi_{+})-\epsilon_{21}]}&,&\mathcal{Z}_{\{(1,1),(1,2),(2,2)\}}\neq\varnothing,\\[1mm]
\tilde{\eta}_{(2,1)}+\frac{\lambda_{21}^{2}-\epsilon_{21}^{2}}{2[\lambda_{21}\cos(\Theta_{+}-\Xi_{+})-\epsilon_{21}]}&,&\mathcal{Z}_{\{(1,1),(2,1),(2,2)\}}\neq\varnothing,
\end{array}
\right.
\end{equation}
where
\begin{equation}\label{eq:chigamthe}
\begin{array}{l}
\Theta_{\pm}=\arccos
\frac{\epsilon_{22}(\lambda_{21}^{2}-\epsilon_{21}^{2})\pm \epsilon_{21}(\lambda_{22}^{2}-\epsilon_{22}^{2})}{\sqrt{\lambda_{21}^{2}(\lambda_{22}^{2}-\epsilon_{22}^{2})^{2}+\lambda_{22}^{2}(\lambda_{21}^{2}-\epsilon_{21}^{2})^{2}\pm 2\lambda_{21}\lambda_{22}(\lambda_{21}^{2}-\epsilon_{21}^{2})(\lambda_{22}^{2}-\epsilon_{22}^{2})\cos\Phi}},\\[2mm]
\Xi_{\pm}=\arccos
\frac{\lambda_{21}(\lambda_{22}^{2}-\epsilon_{22}^{2})\pm \lambda_{22}(\lambda_{21}^{2}-\epsilon_{21}^{2})\cos\Phi}{\sqrt{\lambda_{21}^{2}(\lambda_{22}^{2}-\epsilon_{22}^{2})^{2}+\lambda_{22}^{2}(\lambda_{21}^{2}-\epsilon_{21}^{2})^{2}\pm 2\lambda_{21}\lambda_{22}(\lambda_{21}^{2}-\epsilon_{21}^{2})(\lambda_{22}^{2}-\epsilon_{22}^{2})\cos\Phi}}.
\end{array}
\end{equation}
\item When $\mathcal{Z}\subseteq\mathcal{P}_{\Omega}$, a non-null MEPI measurement exists if and only if $\alpha$ in Eq. \eqref{eq:glpmrpm} is zero.
\end{enumerate}
\end{theorem}

\indent Finally, we mention the location of the single element of $\mathcal{Z}$ characterizing all MEPI measurements for $p_{\rm guess}^{\rm post}>\tilde{\eta}_{\bm{1}}$.
When $\bm{z}\in\mathcal{Z}_{\mathsf{S}}$ for some $\mathsf{S}\subseteq\Omega$,
the definition of $\mathcal{Z}_{\mathsf{S}}$ in Eq.~\eqref{eq:dzss} implies that
two points $\bm{z}$ and $\tilde{\bm{\mu}}_{\bm{\omega}}$ are separated by $p_{\rm guess}^{\rm post}-\tilde{\eta}_{\bm{\omega}}$ for each $\bm{\omega}\in\mathsf{S}$. Thus, the location of $\bm{z}$ is easily determined by the value of $p_{\rm guess}^{\rm post}$ if the affine dimension of $\{\tilde{\bm{\mu}}_{\bm{\omega}}\}_{\bm{\omega}\in\mathsf{S}}$ is one(that is, $\{\tilde{\bm{\mu}}_{\bm{\omega}}\}_{\bm{\omega}\in\mathsf{S}}$ form a line segment).
On the other hand, if the affine dimension of $\{\tilde{\bm{\mu}}_{\bm{\omega}}\}_{\bm{\omega}\in\mathsf{S}}$ is two or three(that is, $\{\tilde{\bm{\mu}}_{\bm{\omega}}\}_{\bm{\omega}\in\mathsf{S}}$ form a polygon or polyhedron), it is difficult to determine the location of $\bm{z}$ only with the value of $p_{\rm guess}^{\rm post}$, and more information about $\bm{z}$ is required.
For example, $\Theta_{\pm}$ and $\Xi_{\pm}$ in Eq.~\eqref{eq:chigamthe} are used to express the specific angles 
that determine $\bm{z}$ when  $m=n_{1}=n_{2}=2$, $p_{\rm guess}^{\rm post}<p_{\rm guess}^{\rm prior}$, and $\mathcal{Z}_{\mathsf{S}}\neq\varnothing$ for some $\mathsf{S}\subseteq\Omega$ with $|\mathsf{S}|=3$.
The details are given in Appendix~\ref{app:thmp}.

\section{Conclusion}\label{sec:conc}
We have analyzed MEPI of qubit state ensemble consisting of $m$ subensembles with $n_{1},\ldots,n_{m}$ states 
by considering the modified problem, that is, ME of $\prod_{b=1}^{m}n_{b}$ qubit states,
and have characterized all optimal measurements.
For the case where MEPI is impossible without the help of measurement, we have 
provided an analytic structure to characterize optimal measurements with fixed non-null outcomes.
We have also shown that a null optimal measurement always exists for any MEPI of qubit states.
In addition, we have analytically provided a necessary and sufficient condition that pre-measurement information is strictly more favorable than post-measurement information(that is, $p_{\rm guess}^{\rm post}<p_{\rm guess}^{\rm prior}$)
when $n_{b}=2$ for all $b=1,\ldots,m$.
For the case of $m=n_{1}=n_{2}=2$ and $p_{\rm guess}^{\rm post}<p_{\rm guess}^{\rm prior}$, we have analytically provided the optimal probability of correct guessing.
In this case, we have also shown that the uniqueness of optimal measurement is equivalent to the non-existence of non-null MEPI measurement. Whereas the existing results of MEPI only deal with some special cases of state ensembles consisting of equiprobable states or symmetric states\cite{ref:ball2008,ref:gopa2010,ref:carm2018}, our results consider the general case of qubit state ensembles having four states and provide the analytic solutions. \\
\indent Our results are closely related to the incompatibility of measurements\cite{ref:carm2018,ref:hein2016,ref:carm2019,ref:skrz20191,ref:uola2019,ref:skrz20192}. 
Recently, it has been found that 
the gap between the two optimal probabilities of correct guessing with pre- and post-measurement information, $p_{\rm guess}^{\rm prior}$ and $p_{\rm guess}^{\rm post}$, is a witness of the incompatibility of measurements\cite{ref:carm2019}.
Moreover, the incompatibility robustness\cite{ref:skrz20191,ref:uola2019,ref:skrz20192} 
of a given set of measurements $\{\mathcal{M}_{b}\}_{b\in\mathsf{B}}$
can be written by $p_{\rm guess}^{\rm post}$ as
\begin{equation}\label{eq:irmb}
I_{\rm R}(\{\mathcal{M}_{b}\}_{b\in\mathsf{B}})=\max_{\mathcal{E}}\frac{p_{\rm corr}^{\rm prior}(\mathcal{E},\{\mathcal{M}_{b}\}_{b\in\mathsf{B}})}{p_{\rm guess}^{\rm post}(\mathcal{E})}-1,
\end{equation}
where the maximization is over all qubit state ensemble $\mathcal{E}$ such as Eq.~\eqref{eq:mens} and
$p_{\rm corr}^{\rm prior}(\mathcal{E},\{\mathcal{M}_{b}\}_{b\in\mathsf{B}})$ is the average probability of correct guessing that can be obtained by performing $\mathcal{M}_{b}$ according to the pre-measurement information $b$. 
As quantum incompatibility is related with quantum resource theory\cite{ref:chit2019}
which is one of the most general frame work in quantum information processing,
it is an interesting future work to obtain an analytic form of Eq.~\eqref{eq:irmb}.

\section*{Acknowledgements}
We would like to thank Professor Joonwoo Bae for his contribution to the early stage of this work. We are grateful to Jaehee Shin and Jihwan Kim for useful comments.
This work is supported by the Basic Science Research Program through
the National Research Foundation of Korea(NRF) funded by the Ministry of
Education, Science and Technology (NRF2018R1D1A1B07049420) and Institute of Information \& communications Technology Planning \& Evaluation(IITP) grant funded by the Korea government(MSIT) (No.2020001343, Artificial Intelligence Convergence Research Center(Hanyang University ERICA).
D.H. and J.S.K. acknowledge support from the National Research Foundation of Korea(NRF)
grant funded by the Korea government(Ministry of Science and ICT)(NRF2020M3E4A1080088).

\appendix
\section{Proofs of Propositions}\label{app:hnome}
\indent In the context of convex optimization\cite{ref:boyd2004}, ME of $\mathcal{E}$ is a convex optimization problem that maximizes $\sum_{i\in\Omega}\eta_{i}{\rm Tr}[\rho_{i}M_{i}]$ subject to the following POVM constraints:
\begin{equation}\label{eq:ppcs}
\begin{array}{ll}
M_{i}\succeq 0\,\forall i\in\Omega,& (\mbox{positive-semidefiniteness})\\[1mm]
\sum_{i\in\Omega}M_{i}=\mathbbm{1}.& (\mbox{completeness})
\end{array}
\end{equation}
The Lagrange dual problem is to minimize ${\rm Tr}K$ subject to 
\begin{equation}\label{eq:ldpcs}
\begin{array}{ll}
W_{i}\succeq 0\,\forall i\in\Omega,& (\mbox{positive-semidefiniteness})\\[1mm]
K=\eta_{i}\rho_{i}+W_{i}\,\forall i\in\Omega, & (\mbox{Lagrangian stability})
\end{array}
\end{equation}
where $W_{i}$ and $K$ are Lagrange multipliers of $M_{i}\succeq 0$ and $\sum_{i\in\Omega}M_{i}=\mathbbm{1}$, respectively\cite{ref:elda20031,ref:bae20132}.\\
\indent As the primal and dual problems have the same optimal value\cite{ref:elda20031}, primal and dual feasible variables are optimal if and only if they satisfy
\begin{equation}\label{eq:cmsn}
{\rm Tr}[M_{i}W_{i}]=0\,\forall i\in\Omega.\ (\mbox{complementary slackness})
\end{equation}
Therefore, ME of $\mathcal{E}$ is equivalent to finding 
a set of primal and dual variables,
$\{M_{i}\}_{i\in\Omega}\cup\{W_{i}\}_{i\in\Omega}\cup\{K\}$, that satisfies the so-called \emph{Karush--Kuhn--Tucker}(KKT) condition consisting of Conditions \eqref{eq:ppcs}, \eqref{eq:ldpcs}, and \eqref{eq:cmsn}\cite{ref:bae20132}.\\
\indent The primal and dual variables can be expressed as
\begin{equation}\label{eq:povmgr}
M_{i}=p_{i}(\mathbbm{1}+\bm{u}_{i}\cdot\bm{\sigma}),\
W_{i}=\frac{r_{i}}{2}(\mathbbm{1}+\bm{w}_{i}\cdot\bm{\sigma}),\ 
K=\frac{1}{2}(s\mathbbm{1}+{\bm v}\cdot\bm{\sigma}).
\end{equation}
using $p_{i},r_{i},s\in\mathbb{R}$ and $\bm{u}_{i},\bm{w}_{i},{\bm v}\in\mathbb{R}^{3}$.
From the expression of \eqref{eq:povmgr}, KKT condition can be expressed as follows:
\begin{equation}\label{eq:mdgeo}
\begin{array}{lll}
\mbox{\bf(P0)}\ p_{i}\geqslant0,\ \ \|\bm{u}_{i}\|\leqslant1\ \ \forall i\in\Omega,\ \,\mbox{\bf(P1)}\ \sum_{i\in\Omega}p_{i}=1,&\mbox{\bf(P2)}\ \sum_{i\in\Omega}p_{i}{\bm u}_{i}={\bm 0},\\
&&\\
\mbox{\bf(D0)}\ r_{i}\geqslant0,\ \ \|\bm{w}_{i}\|\leqslant1\ \ \forall i\in\Omega,\ \mbox{\bf(D1)}\ s=\eta_{i}+r_{i}\ \forall i\in\Omega,&\mbox{\bf(D2)}\ {\bm v}=\eta_{i}{\bm \nu}_{i}+r_{i}{\bm w}_{i}\ \forall i\in\Omega,\\
&&\\
\mbox{\bf(C0)}\ p_{i}r_{i}=0\ \ \mbox{or}\ \ {\bm u}_{i}\cdot{\bm w}_{i}=-1\ \ \forall i\in\Omega.&
\end{array}
\end{equation}
\indent Now, we prove Propositions~\ref{pro:nomec} and  \ref{pro:htp} by using KKT condition \eqref{eq:mdgeo}
with $\{p_{i},{\bm u}_{i}\}_{i\in\Omega}\cup\{r_{i},{\bm w}_{i}\}_{i\in\Omega}\cup\{s,{\bm v}\}$ as variables.
We use the superscript ${}^\star$ to express the optimal variables.
Note that $s^{\star}$, ${\bm v}^{\star}$, and $r_{i}^{\star}$ are always unique, whereas $p_{i}^{\star}$, $\bm{u}_{i}^{\star}$, and $\bm{w}_{i}^{\star}$ may not\cite{ref:bae20132}.
We also note that $s^{\star}=p_{\rm guess}$.

\subsection*{Proof of Proposition~\ref{pro:nomec}}
\renewcommand*{\proofname}{Proof of {\rm(a)}}
\begin{proof}
For the proof of Proposition~\ref{pro:nomec}{\rm(a)}($\Rightarrow$), assume $p_{\rm guess}=\eta_{1}$. 
From {\bf(D1)} and {\bf(D2)}  of \eqref{eq:mdgeo},
this assumption implies $\eta_{1}>\eta_{i}$ for all $i\neq1$
because $\eta_{j}=\eta_{1}$ means $\bm{\nu}_{j}=\bm{\nu}_{1}$. Moreover,
\begin{equation}
r_{i}^{\star}=\eta_{1}-\eta_{i},\, \ \bm{w}_{i}^{\star}=\frac{\eta_{1}\bm{\nu}_{1}-\eta_{i}\bm{\nu}_{i}}{\eta_{1}-\eta_{i}}\ \ \forall i\neq1.
\end{equation}
Thus,
\begin{equation}
\begin{array}{rcl}
\epsilon_{i}=\eta_{1}-\eta_{i}\geqslant
\|\bm{w}_{i}^{\star}\|(\eta_{1}-\eta_{i})
=\|\eta_{1}\bm{\nu}_{1}-\eta_{i}\bm{\nu}_{i}\|=\lambda_{i}\ \forall i\in\Omega,
\end{array}
\end{equation}
where the inequality is due to {\bf(D0)}  of \eqref{eq:mdgeo}.\\
\indent For the proof of Proposition~\ref{pro:nomec}{\rm(a)}($\Leftarrow$), suppose that $\epsilon_{i}\geqslant\lambda_{i}$ for all $i\in\Omega$.
Since $\eta_{j}=\eta_{1}$ means $\bm{\nu}_{j}=\bm{\nu}_{1}$ by $\epsilon_{j}\geqslant\lambda_{j}$,
it follows that $\eta_{1}>\eta_{i}$ for all $i\neq1$.
Thus, $p_{\rm guess}=\eta_{1}$ because \eqref{eq:mdgeo} holds for 
\begin{equation}\label{eq:pidui0}
\begin{array}{lll}
p_{i}=\delta_{1i},&\bm{u}_{i}=\bm{0},\\[1mm]
r_{i}=\eta_{1}-\eta_{i},&\bm{w}_{i}=
\Bigg\{
\begin{array}{ccl}
\bm{0}&\mbox{for}&i=1,\\
\frac{\eta_{1}\bm{\nu}_{1}-\eta_{i}\bm{\nu}_{i}}{\eta_{1}-\eta_{i}}&\mbox{for}&i\neq1,
\end{array}
\\[1mm]
s=\eta_{1},&\bm{v}=\eta_{1}\bm{\nu}_{1},
\end{array}
\end{equation}
where $\delta_{ij}$ is the Kronecker delta. 
This completes our proof of Proposition~\ref{pro:nomec}(a).
We also note that
$\epsilon_{i}\geqslant\lambda_{i}$ for all $i\in\Omega$
is equivalent to $\mathcal{Z}_{\{1\}}\neq\varnothing$ since $\mathcal{P}_{\{1\}}$ has only one element $\eta_{1}\bm{\nu}_{1}$.
\end{proof}
\renewcommand*{\proofname}{Proof}
\renewcommand*{\proofname}{Proof of {\rm(b)}}
\begin{proof}
For the proof of Proposition~\ref{pro:nomec}{\rm(b)}($\Rightarrow$), we assume that $p_{\rm guess}=\eta_{1}$ and $\mathcal{M}$ is an optimal POVM expressed by $\{p_{i},\bm{u}_{i}\}_{i\in\Omega}$. 
Then, there is $\{r_{i},\bm{w}_{i}\}_{i\in\Omega}\cup\{s,{\bm v}\}$ such that
\eqref{eq:mdgeo} hold for $\{p_{i},\bm{u}_{i}\}_{i\in\Omega}\cup\{r_{i},\bm{w}_{i}\}_{i\in\Omega}\cup\{s,{\bm v}\}$.\\
\indent Let us consider a non-null outcome $k\neq1$.
Since $s=p_{\rm guess}=\eta_{1}$, it follows from {\bf(D1)} and {\bf(D2)}  of \eqref{eq:mdgeo} that 
\begin{equation}
r_{k}=\eta_{1}-\eta_{k},\ r_{k}\bm{w}_{k}=\eta_{1}\bm{\nu}_{1}-\eta_{k}\bm{\nu}_{k}.
\end{equation}
$r_{k}$ is nonzero because $\eta_{k}=\eta_{1}$ means ${\bm \nu}_{k}={\bm \nu}_{1}$ by $\epsilon_{k}\geqslant\lambda_{k}$, therefore,
\begin{equation}
\bm{w}_{k}=\frac{\eta_{1}\bm{\nu}_{1}-\eta_{k}\bm{\nu}_{k}}{\eta_{1}-\eta_{k}}.
\end{equation}
Since $p_{k}$ and $r_{k}$ are nonzero, it follows from {\bf(C0)} of \eqref{eq:mdgeo}  that $\bm{u}_{k}$ and $\bm{w}_{k}$ are unit vectors satisfying $\bm{u}_{k}\cdot\bm{w}_{k}=-1$.
Thus,
\begin{equation}
\epsilon_{k}=\epsilon_{k}\|{\bm w}_{k}\|=\lambda_{k}
\end{equation}
and
\begin{equation}
\bm{u}_{k}=\frac{\eta_{k}\bm{\nu}_{k}-\eta_{1}\bm{\nu}_{1}}{\|\eta_{k}\bm{\nu}_{k}-\eta_{1}\bm{\nu}_{1}\|}, 
\end{equation}
that is, Eq.~\eqref{eq:trime} holds  for $i=k$. 
We also note that Eq.~\eqref{eq:trime} is obviously satisfied for every null outcome.
Therefore, $\mathcal{M}$ satisfies Eq.~\eqref{eq:trime}.\\ 
\indent For the proof of Proposition~\ref{pro:nomec}{\rm(b)}($\Leftarrow$), suppose that $p_{\rm guess}=\eta_{1}$ and $\mathcal{M}$ is a POVM satisfying Eq.~\eqref{eq:trime}.
Let $\{p_{i},\bm{u}_{i}\}_{i\in\Omega}$ be a set of primal variables expressing $\mathcal{M}$.
Also, let $\{r_{i},\bm{w}_{i}\}_{i\in\Omega}\cup\{s,\bm{v}\}$ be a set of dual variables such that
\begin{equation}
\begin{array}{ll}
r_{i}=\eta_{1}-\eta_{i},&\bm{w}_{i}=
\Bigg\{
\begin{array}{ccl}
\bm{0}&\mbox{for}&i=1,\\
\frac{\eta_{1}\bm{\nu}_{1}-\eta_{i}\bm{\nu}_{i}}{\eta_{1}-\eta_{i}}&\mbox{for}&i\neq1,
\end{array}
\\
s=\eta_{1},& \bm{v}=\eta_{1}\bm{v}_{1}.
\end{array}
\end{equation}
Then, \eqref{eq:mdgeo} holds for $\{p_{i},\bm{u}_{i}\}_{i\in\Omega}\cup\{r_{i},\bm{w}_{i}\}_{i\in\Omega}\cup\{s,\bm{v}\}$.
Therefore, $\mathcal{M}$ is optimal.
\end{proof}
\renewcommand*{\proofname}{Proof}

\subsection*{Proof of Proposition~\ref{pro:htp}}
\renewcommand*{\proofname}{Proof of {\rm (a)}, {\rm (b)($\Leftarrow$)}, and {\rm (c)($\Leftarrow$)}}
\begin{proof}
Let us assume that $\mathsf{S}$ is a subset of $\Omega$ with $\bm{v}\in\mathcal{Z}_{\mathsf{S}}$. The definition of $\mathcal{Z}_{\mathsf{S}}$ implies that there exists a real number $s$ fulfilling
$\varphi_{i}(\bm{v})=s$ for all $i\in\mathsf{S}$. 
We will show that $s=s^{\star}$ and ${\bm v}={\bm v}^{\star}$.
Since $s^{\star}=p_{\rm guess}$ and $\bm{v}^{\star}$ is unique, 
if $s=s^{\star}$ and ${\bm v}={\bm v}^{\star}$, 
then $p_{\rm guess}=\varphi_{i}(\bm{v})$ for all $i\in\mathsf{S}$
and $\mathcal{Z}_{\mathsf{S}}=\{\bm{v}^{\star}\}$ which implies that $\mathcal{Z}$ is a single-element set having only one element $\bm{v}$, that is, $\mathcal{Z}=\{\bm{v}\}$.
\\
\indent First, let us consider the case of $\eta_{k}=s$ for some $k\in\Omega$.
Since $s\geqslant\varphi_{i}(\bm{v})$ for all $i\in\Omega$ from the definition of $\mathcal{Z}_{\mathsf{S}}$ and $\varphi_{j}(\bm{v})=s$ for all $j\in\mathsf{S}$, we have
\begin{equation}
\|\eta_{k}\bm{\nu}_{k}-\bm{v}\|=\varphi_{k}(\bm{v})-\eta_{k}\leqslant s-\eta_{k}=0,
\end{equation}
which means $\bm{v}=\eta_{k}\bm{\nu}_{k}$. Thus, $k=1$ because
\begin{equation}
\eta_{k}=s\geqslant\varphi_{i}(\eta_{k}\bm{\nu}_{k})=\eta_{i}+\|\eta_{i}\bm{\nu}_{i}-\eta_{k}\bm{\nu}_{k}\|>\eta_{i}\ \forall i\neq k.
\end{equation}
\eqref{eq:mdgeo} holds for the primal and dual variables satisfying Eq.~\eqref{eq:pidui0}; therefore, $s=s^{\star}$ and ${\bm v}={\bm v}^{\star}$.\\
\indent Now, let us consider the case of $\eta_{i}\neq s$ for all $i\in\Omega$.
In this case, $\eta_{i}\bm{\nu}_{i}\neq\bm{v}$ for all $i\in\mathsf{S}$
because otherwise 
\begin{equation}
s=\varphi_{j}(\bm{v})=\eta_{j}+\|\eta_{j}\bm{\nu}_{j}-\bm{v}\|=\eta_{j}
\end{equation}
for some $j\in\mathsf{S}$ with $\eta_{j}\bm{\nu}_{j}=\bm{v}$.
Since ${\bm v}\in\mathcal{P}_{\mathsf{S}}$, there exists $\{c_{i}\}_{i\in\Omega}\subseteq\mathbb{R}$ such that 
\begin{equation}
c_{i}>0\ \forall i\in\mathsf{S},\ \
c_{i}=0\ \forall i\not\in\mathsf{S},\ \
\sum_{i\in\mathsf{S}}c_{i}=1,\ \
\sum_{i\in\mathsf{S}}c_{i}(\eta_{i}\bm{\nu}_{i})={\bm v}.
\end{equation}
Thus, $s=s^{\star}$ and ${\bm v}={\bm v}^{\star}$ because \eqref{eq:mdgeo} holds for $\{s,{\bm v}\}$ along with $\{p_{i},\bm{u}_{i}\}_{i\in\Omega}\cup\{r_{i},\bm{w}_{i}\}_{i\in\Omega}$ satisfying
\begin{equation}\label{eq:pruwi}
\begin{array}{rclrcl}
p_{i}&=&
\frac{c_{i}\|\eta_{i}\bm{\nu}_{i}-{\bm v}\|}{\sum_{j\in\mathsf{S}}c_{j}\|\eta_{j}\bm{\nu}_{j}-{\bm v}\|},
&
\quad r_{i}&=&s-\eta_{i},
\\
\bm{u}_{i}&=&
\left\{
\begin{array}{ccl}
\frac{\eta_{i}\bm{\nu}_{i}-{\bm v}}{\|\eta_{i}\bm{\nu}_{i}-{\bm v}\|}&,&i\in\mathsf{S},\\
\bm{0}&,&i\not\in\mathsf{S},
\end{array}
\right.
&\bm{w}_{i}&=&
\left\{
\begin{array}{ccl}
\frac{{\bm v}-\eta_{i}\bm{\nu}_{i}}{\|{\bm v}-\eta_{i}\bm{\nu}_{i}\|}&,&i\in\mathsf{S},\\
\frac{{\bm v}-\eta_{i}\bm{\nu}_{i}}{s-\eta_{i}}&,&i\not\in\mathsf{S}.
\end{array}
\right.
\end{array}
\end{equation}
This completes our proof of Proposition~\ref{pro:htp}(a).\\
\indent Moreover, the optimal POVM corresponding to $\{p_{i},\bm{u}_{i}\}_{i\in\Omega}$ satisfying Eq.~\eqref{eq:pruwi} has $\mathsf{S}$ as the set of all non-null outcomes.
Thus, Proposition~\ref{pro:htp}(b)($\Leftarrow$) is proved. We also note that
every POVM $\mathcal{M}$ fulfilling Condition \eqref{eq:mem} with $M_{i}=0$ for all $i\notin\mathsf{S}$ is optimal because \eqref{eq:mdgeo} holds for
$\{s,\bm{v}\}$ along with $\{p_{i},\bm{u}_{i}\}_{i\in\Omega}$ expressing $\mathcal{M}$ 
and $\{r_{i},\bm{w}_{i}\}_{i\in\Omega}$ satisfying \eqref{eq:pruwi};
therefore, Proposition~\ref{pro:htp}(c)($\Leftarrow$) is also proved.
\end{proof}
\renewcommand*{\proofname}{Proof}
\renewcommand*{\proofname}{Proof of {\rm(b)($\Rightarrow$)} and {\rm(c)($\Rightarrow$)}}
\begin{proof}
Suppose that $p_{\rm guess}>\eta_{1}$ and $\mathcal{M}$ is an optimal POVM for ME of $\mathcal{E}$ having $\mathsf{S}$ as the set of all non-null outcomes.
We also assume that $\{p_{i},\bm{u}_{i}\}_{i\in\Omega}$ is a set of primal variables expressing $\mathcal{M}$. Then, there exists a set of dual variables $\{r_{i},\bm{w}_{i}\}_{i\in\Omega}\cup\{s,{\bm v}\}$ such that \eqref{eq:mdgeo} holds for 
$\{p_{i},\bm{u}_{i}\}_{i\in\Omega}\cup\{r_{i},\bm{w}_{i}\}_{i\in\Omega}\cup\{s,{\bm v}\}$.
Since $p_{\rm guess}>\eta_{1}$, it follows from {\bf(D1)} of \eqref{eq:mdgeo} that $r_{i}$ is nonzero for all $i\in\Omega$. 
Furthermore, {\bf(C0)} of \eqref{eq:mdgeo}
implies $\|\bm{u}_{j}\|=\|\bm{w}_{j}\|=1$ and $\bm{u}_{j}\cdot\bm{w}_{j}=-1$ for each $j\in\mathsf{S}$ because both $p_{j}$ and $r_{j}$ are nonzero.
Thus, from {\bf(D2)} of \eqref{eq:mdgeo}, we have
\begin{equation}\label{eq:rjuj}
\begin{array}{ll}
r_{j}=\|\eta_{j}\bm{\nu}_{j}-{\bm v}\|,\ \bm{u}_{j}=\frac{\eta_{j}\bm{\nu}_{j}-{\bm v}}{\|\eta_{j}\bm{\nu}_{j}-{\bm v}\|}\ \forall j\in\mathsf{S}.
\end{array}
\end{equation}
The right equality in \eqref{eq:rjuj} is equivalent to Condition \eqref{eq:mem};
therefore, Proposition~\ref{pro:htp}(c)($\Rightarrow$) is proved. \\
\indent Now, we show $\mathcal{Z}_{\mathsf{S}}\neq\varnothing$ to prove Proposition~\ref{pro:htp}(b)($\Rightarrow$).
From {\bf(D0)} and {\bf(D2)} of \eqref{eq:mdgeo}, we can see that
\begin{equation}\label{eq:ririwi}
r_{i}\geqslant\|r_{i}{\bm w}_{i}\|=\|\eta_{i}\bm{\nu}_{i}-{\bm v}\|\ \forall i\in\Omega.
\end{equation}
Thus,
\begin{equation}\label{eq:vpvph}
\begin{array}{rcl}
\varphi_{i}({\bm v})&=&q_{i}+\|\eta_{i}\bm{\nu}_{i}-{\bm v}\|
=q_{i}+r_{i}=q_{j}+r_{j}\\
&\geqslant&q_{j}+\|\eta_{j}\bm{\nu}_{j}-{\bm v}\|=\varphi_{j}({\bm v})\ \forall i\in\mathsf{S},\, j\in\Omega,
\end{array}
\end{equation}
where the second equality is due to Eq.~\eqref{eq:rjuj} and
the third equality is from {\bf(D1)} of \eqref{eq:mdgeo}.
Since the inequality in \eqref{eq:vpvph} for $j\in\mathsf{S}$ becomes an equality by Eq.~\eqref{eq:rjuj},
$\mathcal{Z}_{\mathsf{S}}$ is a nonempty set with ${\bm v}$. Note that  ${\bm v}\in\mathcal{P}_{\mathsf{S}}$ because
\begin{equation}\label{eq:cietiv}
\sum_{i\in\mathsf{S}}c_{i}(\eta_{i}\bm{\nu}_{i})=
\sum_{i\in\mathsf{S}}c_{i}(\bm{v}-r_{i}\bm{w}_{i})
=\bm{v}-\frac{\sum_{i\in\mathsf{S}}(p_{i}\bm{w}_{i})}{\sum_{j\in\mathsf{S}}(p_{j}/r_{j})}
=\bm{v}+\frac{\sum_{i\in\mathsf{S}}(p_{i}\bm{u}_{i})}{\sum_{j\in\mathsf{S}}(p_{j}/r_{j})}={\bm v},
\end{equation}
where 
\begin{equation}
c_{i}=\frac{(p_{i}/r_{i})}{\sum_{j\in\mathsf{S}}(p_{j}/r_{j})}.
\end{equation}
The first equality of \eqref{eq:cietiv} is due to {\bf(D2)} of \eqref{eq:mdgeo}, 
the second equality is from $\sum_{i\in\mathsf{S}}c_{i}=1$, 
and the third and last equality follow from {\bf(C0)} and {\bf(P2)}, respectively.
Therefore, Proposition~\ref{pro:htp}(b)($\Rightarrow$) is proved.
\end{proof}
\renewcommand*{\proofname}{Proof}

\section{Proofs of Theorems}\label{app:thmp}
In this section, we prove Theorems~\ref{thm:twom} and \ref{thm:prfthr} using $\mathcal{X}_{\mathsf{S}'}^{\mathsf{S}}$ and $\mathcal{Y}_{\mathsf{S}'}^{\mathsf{S}}$($\mathsf{S}'\subseteq\mathsf{S}\subseteq\Omega$) defined as
\begin{equation}\label{eq:xysd}
\begin{array}{rcl}
\mathcal{X}_{\mathsf{S}'}^{\mathsf{S}}:=\{\bm{v}\in\mathcal{P}_{\mathsf{S}'}&:&\varphi_{\bm{\omega}}(\bm{v})=\varphi_{\bm{\omega}'}(\bm{v})\ \forall(\bm{\omega},\bm{\omega}')\in\mathsf{S}'\times\mathsf{S}',\\
&&\varphi_{\bm{\omega}}(\bm{v})\geqslant\varphi_{\bm{\omega}'}(\bm{v})\ \forall(\bm{\omega},\bm{\omega}')\in\mathsf{S}'\times(\mathsf{S}-\mathsf{S}')\},\\[2mm]
\mathcal{Y}_{\mathsf{S}'}^{\mathsf{S}}:=\{{\bm v}\in\mathbb{R}^{3}\ &:&
\varphi_{\bm{\omega}}(\bm{v})\geqslant\varphi_{\bm{\omega}'}(\bm{v})\ \forall(\bm{\omega},\bm{\omega}')\in\mathsf{S}'\times(\Omega-\mathsf{S})\},
\end{array}
\end{equation}
where we denote 
\begin{equation}\label{eq:xyssx}
\mathcal{X}_{\mathsf{S}}:=\mathcal{X}_{\mathsf{S}}^{\mathsf{S}},\
\mathcal{Y}_{\mathsf{S}}:=\mathcal{Y}_{\mathsf{S}}^{\mathsf{S}}.
\end{equation}
Note that $\mathcal{X}_{\mathsf{S}'}^{\mathsf{S}}\cap\mathcal{Y}_{\mathsf{S}'}^{\mathsf{S}}=\mathcal{Z}_{\mathsf{S}'}$ for any $\mathsf{S}'\subseteq\mathsf{S}\subseteq\Omega$.
We also note that $\bigcup_{\mathsf{S}'\subseteq\mathsf{S}}
\mathcal{X}_{\mathsf{S}'}^{\mathsf{S}}$
is a single-element set for any $\mathsf{S}\subseteq\Omega$ in terms of ME of $\{\tilde{\eta}_{\bm{\omega}},\tilde{\rho}_{\bm{\omega}}\}_{\bm{\omega}\in\mathsf{S}}$; thus, at least one of $\mathcal{X}_{\mathsf{S}'}^{\mathsf{S}}$ and $\mathcal{X}_{\mathsf{S}''}^{\mathsf{S}}$ is empty 
for any $\mathsf{S}',\mathsf{S}''\subseteq\mathsf{S}$ with $\mathcal{P}_{\mathsf{S}'}\cap\mathcal{P}_{\mathsf{S}''}=\varnothing$. In particular, we use the following two lemmas 
when proving Theorem~\ref{thm:prfthr}.

\begin{lemma}\label{lem:xnep}
When $m=n_{1}=n_{2}=2$ and $p_{\rm guess}^{\rm post}<p_{\rm guess}^{\rm prior}$, if $\mathsf{S}$ is a proper subset of $\Omega$ with $\mathcal{X}_{\mathsf{S}}\neq\varnothing$
and $\mathsf{S}'$ is a proper subset of $\mathsf{S}$
satisfying $\mathsf{S}-\mathsf{S}'=\{\bm{\omega}\}$ and $\mathcal{X}_{\mathsf{S}'}=\{\bm{v}\}$, then
\begin{equation}\label{eq:vwpv}
\begin{array}{lclcl}
\varphi_{\bm{\omega}'}(\bm{v})<\varphi_{\bm{\omega}}(\bm{v})\ \forall\bm{\omega}'\in\mathsf{S}'.
\end{array}
\end{equation}
\end{lemma}
\begin{proof}
$\mathcal{P}_{\mathsf{S}}$ and $\mathcal{P}_{\mathsf{S}'}$ are disjoint for any $\mathsf{S'}\subsetneq\mathsf{S}\subsetneq\Omega$ because $\{\tilde{\bm{\mu}}_{\bm{\omega}}\}_{\bm{\omega}\in\Omega}$
forms a parallelogram with nonempty interior
when $m=n_{1}=n_{2}=2$ and $p_{\rm guess}^{\rm post}<p_{\rm guess}^{\rm prior}$.
Thus, if $\mathcal{X}_{\mathsf{S}}$ is nonempty for some $\mathsf{S}\subsetneq\Omega$, then 
$\mathcal{X}_{\mathsf{S}'}^{\mathsf{S}}$ is empty
for all $\mathsf{S}'\subsetneq\mathsf{S}$.
Now, we assume that $\mathsf{S}$ is a proper subset of $\Omega$ with $\mathcal{X}_{\mathsf{S}}\neq\varnothing$
and $\mathsf{S}'$ is a proper subset of $\mathsf{S}$
satisfying $\mathsf{S}-\mathsf{S}'=\{\bm{\omega}\}$ and $\mathcal{X}_{\mathsf{S}'}=\{\bm{v}\}$. If $\varphi_{\bm{\omega}'}(\bm{v})\geqslant\varphi_{\bm{\omega}}(\bm{v})$ for some $\bm{\omega}'\in\mathsf{S}'$, then $\mathcal{X}_{\mathsf{S}'}^{\mathsf{S}}$ becomes a nonempty set with $\bm{v}$ as a element.
Thus, $\mathcal{X}_{\mathsf{S}'}^{\mathsf{S}}=\varnothing$ implies \eqref{eq:vwpv}, which completes our proof.
\end{proof}

\begin{lemma}\label{lem:sufc}
When $m=n_{1}=n_{2}=2$ and $p_{\rm guess}^{\rm post}<p_{\rm guess}^{\rm prior}$, $\mathcal{X}_{\mathsf{S}}$ is nonempty for $\mathsf{S}\subsetneq\Omega$
if $\mathcal{X}_{\mathsf{S}'}\neq\varnothing$ and $\mathcal{X}_{\mathsf{S}'}^{\mathsf{S}}=\varnothing$ for all $\mathsf{S}'\subsetneq\mathsf{S}$ with $|\mathsf{S}-\mathsf{S}'|=1$.
\end{lemma}
\begin{proof}
For each $\mathsf{S}''\subsetneq\mathsf{S}'\subsetneq\mathsf{S}$ with $|\mathsf{S}-\mathsf{S}'|=1$, $\mathcal{X}_{\mathsf{S}''}^{\mathsf{S}'}$ is empty because $\bigcup_{\mathsf{S}''\subseteq\mathsf{S}'}
\mathcal{X}_{\mathsf{S}''}^{\mathsf{S}'}$ is a single-element set and $\mathcal{X}_{\mathsf{S}'}$ is a nonempty set.
Note that $\mathcal{P}_{\mathsf{S}'}$ and $\mathcal{P}_{\mathsf{S}''}$ are disjoint for any $\mathsf{S}''\subsetneq\mathsf{S}'\subsetneq\mathsf{S}$ since $\{\tilde{\bm{\mu}}_{\bm{\omega}}\}_{\bm{\omega}\in\Omega}$
forms a parallelogram with nonempty interior
when $m=n_{1}=n_{2}=2$ and $p_{\rm guess}^{\rm post}<p_{\rm guess}^{\rm prior}$. For $\mathsf{S}''\subsetneq\mathsf{S}'\subsetneq\mathsf{S}$,
emptiness of $\mathcal{X}_{\mathsf{S}''}^{\mathsf{S}'}$ implies
emptiness of $\mathcal{X}_{\mathsf{S}''}^{\mathsf{S}}$
as $\mathcal{X}_{\mathsf{S}''}^{\mathsf{S}}\subseteq\mathcal{X}_{\mathsf{S}''}^{\mathsf{S}'}$; therefore, $\mathcal{X}_{\mathsf{S}'}^{\mathsf{S}}$ is empty for all $\mathsf{S}'\subsetneq\mathsf{S}$.
Since $\bigcup_{\mathsf{S}'\subseteq\mathsf{S}}
\mathcal{X}_{\mathsf{S}'}^{\mathsf{S}}$ is a single-element set,
$\mathcal{X}_{\mathsf{S}}$ is nonempty if 
$\mathcal{X}_{\mathsf{S}'}^{\mathsf{S}}$ is empty for all $\mathsf{S}'\subsetneq\mathsf{S}$.
Thus, $\mathcal{X}_{\mathsf{S}}$ is nonempty.
\end{proof}

\subsection*{Proof of Theorem~\ref{thm:twom}}
\begin{proof}
Suppose that $m=n_{1}=n_{2}=2$ and $p_{\rm guess}^{\rm post}<p_{\rm guess}^{\rm prior}$.
We can easily verify that
\begin{equation}\label{eq:xvvv}
\begin{array}{lcl}
\mathcal{X}_{\{(1,1),(1,2)\}}=\{{\bm v}_{11}\},\ \ \varphi_{(1,1)}(\bm{v}_{11})+\varphi_{(1,2)}(\bm{v}_{11})=2\eta_{11}+\eta_{12}+\eta_{22}+\lambda_{22},\\[1.5mm]
\mathcal{X}_{\{(1,1),(2,1)\}}=\{{\bm v}_{12}\},\ \ \varphi_{(1,1)}(\bm{v}_{12})+\varphi_{(2,1)}(\bm{v}_{12})=2\eta_{12}+\eta_{11}+\eta_{21}+\lambda_{21},\\[1.5mm]
\mathcal{X}_{\{(2,1),(2,2)\}}=\{{\bm v}_{21}\},\ \ \varphi_{(2,1)}(\bm{v}_{21})+\varphi_{(2,2)}(\bm{v}_{21})=2\eta_{21}+\eta_{12}+\eta_{22}+\lambda_{22},\\[1.5mm]
\mathcal{X}_{\{(1,2),(2,2)\}}=\{{\bm v}_{22}\},\ \ \varphi_{(1,2)}(\bm{v}_{22})+\varphi_{(2,2)}(\bm{v}_{22})=2\eta_{22}+\eta_{11}+\eta_{21}+\lambda_{21},
\end{array}
\end{equation}
where
\begin{equation}
\begin{array}{ll}
{\bm v}_{11}
=\Big(\frac{\lambda_{22}+\epsilon_{22}}{2\lambda_{22}}\Big)\tilde{\bm{\mu}}_{(1,1)}+\Big(\frac{\lambda_{22}-\epsilon_{22}}{2\lambda_{22}}\Big)\tilde{\bm{\mu}}_{(1,2)},&
{\bm v}_{12}
=\Big(\frac{\lambda_{21}+\epsilon_{21}}{2\lambda_{21}}\Big)\tilde{\bm{\mu}}_{(1,1)}+\Big(\frac{\lambda_{21}-\epsilon_{21}}{2\lambda_{21}}\Big)\tilde{\bm{\mu}}_{(2,1)},\\[1.5mm]
{\bm v}_{21}
=\Big(\frac{\lambda_{22}+\epsilon_{22}}{2\lambda_{22}}\Big)\tilde{\bm{\mu}}_{(2,1)}+\Big(\frac{\lambda_{22}-\epsilon_{22}}{2\lambda_{22}}\Big)\tilde{\bm{\mu}}_{(2,2)},&
{\bm v}_{22}
=\Big(\frac{\lambda_{21}+\epsilon_{21}}{2\lambda_{21}}\Big)\tilde{\bm{\mu}}_{(1,2)}+\Big(\frac{\lambda_{21}-\epsilon_{21}}{2\lambda_{21}}\Big)\tilde{\bm{\mu}}_{(2,2)}.
\end{array}
\end{equation}
Since $\mathcal{X}_{\mathsf{S}}\cap\mathcal{Y}_{\mathsf{S}}=\mathcal{Z}_{\mathsf{S}}$ for all $\mathsf{S}\subseteq\Omega$, it follows from Eq.~\eqref{eq:xvvv} that
\begin{equation}
\begin{array}{lclcl}
\mathcal{Z}_{\{(1,1),(1,2)\}}\neq\varnothing&\ \Leftrightarrow\ &
{\bm v}_{11}\in\mathcal{Y}_{\{(1,1),(1,2)\}},\\[1mm]
\mathcal{Z}_{\{(1,1),(2,1)\}}\neq\varnothing&\ \Leftrightarrow\ &
{\bm v}_{12}\in\mathcal{Y}_{\{(1,1),(2,1)\}},\\[1mm]
\mathcal{Z}_{\{(2,1),(2,2)\}}\neq\varnothing&\ \Leftrightarrow\ &
{\bm v}_{21}\in\mathcal{Y}_{\{(2,1),(2,2)\}},\\[1mm]
\mathcal{Z}_{\{(1,2),(2,2)\}}\neq\varnothing&\ \Leftrightarrow\ &
{\bm v}_{22}\in\mathcal{Y}_{\{(1,2),(2,2)\}},
\end{array}
\end{equation}
and
\begin{equation}
p_{\rm guess}^{\rm post}=\left\{
\begin{array}{cc}
\eta_{11}+\frac{1}{2}(\eta_{12}+\eta_{22}+\lambda_{22}),&\mathcal{Z}_{\{(1,1),(1,2)\}}\neq\varnothing,\\[1mm]
\eta_{12}+\frac{1}{2}(\eta_{11}+\eta_{21}+\lambda_{21}),&\mathcal{Z}_{\{(1,1),(2,1)\}}\neq\varnothing,\\[1mm]
\eta_{21}+\frac{1}{2}(\eta_{12}+\eta_{22}+\lambda_{22}),&\mathcal{Z}_{\{(2,1),(2,2)\}}\neq\varnothing,\\[1mm]
\eta_{22}+\frac{1}{2}(\eta_{11}+\eta_{21}+\lambda_{21}),&\mathcal{Z}_{\{(1,2),(2,2)\}}\neq\varnothing.\\[1mm]
\end{array}
\right.
\end{equation}
\indent From definition of $\mathcal{Y}_{\mathsf{S}}$, we can see that
\begin{equation}
\begin{array}{lcr}
\bm{v}_{11}\in\mathcal{Y}_{\{(1,1),(1,2)\}} &\ \Leftrightarrow\ 
&\varphi_{(1,1)}({\bm v}_{11})\geqslant\varphi_{(2,1)}({\bm v}_{11}),\ \ 
\varphi_{(1,2)}({\bm v}_{11})\geqslant\varphi_{(2,2)}({\bm v}_{11}),\\[1.5mm]
\bm{v}_{12}\in\mathcal{Y}_{\{(1,1),(2,1)\}} &\ \Leftrightarrow\ 
&\varphi_{(1,1)}({\bm v}_{12})\geqslant\varphi_{(1,2)}({\bm v}_{12}),\ \ 
\varphi_{(2,1)}({\bm v}_{12})\geqslant\varphi_{(2,2)}({\bm v}_{12}),\\[1.5mm]
\bm{v}_{21}\in\mathcal{Y}_{\{(2,1),(2,2)\}} &\ \Leftrightarrow\ 
&\varphi_{(1,1)}({\bm v}_{21})\leqslant\varphi_{(2,1)}({\bm v}_{21}),\ \ 
\varphi_{(1,2)}({\bm v}_{21})\leqslant\varphi_{(2,2)}({\bm v}_{21}),\\[1.5mm]
\bm{v}_{22}\in\mathcal{Y}_{\{(1,2),(2,2)\}} &\ \Leftrightarrow\ 
&\varphi_{(1,1)}({\bm v}_{22})\leqslant\varphi_{(1,2)}({\bm v}_{22}),\ \ 
\varphi_{(2,1)}({\bm v}_{22})\leqslant\varphi_{(2,2)}({\bm v}_{22}).\\[1.5mm]
\end{array}
\end{equation}
It is straightforward to verify that
\begin{equation}\label{eq:vvvvc2}
\begin{array}{rcl}
\varphi_{(1,1)}({\bm v}_{11})-\varphi_{(2,1)}({\bm v}_{11})
&=&\frac{(\epsilon_{21}+\lambda_{21}\cos\Phi)(\lambda_{22}-\epsilon_{22})-(\lambda_{21}^{2}-\epsilon_{21}^{2})}{\epsilon_{21}+\|\tilde{\bm{\mu}}_{(1,1)}-\bm{v}_{11}\|+\|\tilde{\bm{\mu}}_{(1,2)}-\bm{v}_{11}\|},\\[1.5mm]
\varphi_{(1,2)}({\bm v}_{11})-\varphi_{(2,2)}({\bm v}_{11})
&=&\frac{(\epsilon_{21}-\lambda_{21}\cos\Phi)(\lambda_{22}+\epsilon_{22})-(\lambda_{21}^{2}-\epsilon_{21}^{2})}{\epsilon_{21}+\|\tilde{\bm{\mu}}_{(1,2)}-\bm{v}_{11}\|+\|\tilde{\bm{\mu}}_{(2,2)}-\bm{v}_{11}\|},\\[3mm]
\varphi_{(1,1)}({\bm v}_{12})-\varphi_{(1,2)}({\bm v}_{12})
&=&\frac{(\epsilon_{22}+\lambda_{22}\cos\Phi)(\lambda_{21}-\epsilon_{21})-(\lambda_{22}^{2}-\epsilon_{22}^{2})}{\epsilon_{22}+\|\tilde{\bm{\mu}}_{(1,1)}-\bm{v}_{12}\|+\|\tilde{\bm{\mu}}_{(1,2)}-\bm{v}_{12}\|},\\[1.5mm]
\varphi_{(2,1)}({\bm v}_{12})-\varphi_{(2,2)}({\bm v}_{12})
&=&\frac{(\epsilon_{22}-\lambda_{22}\cos\Phi)(\lambda_{21}+\epsilon_{21})-(\lambda_{22}^{2}-\epsilon_{22}^{2})}{\epsilon_{22}+\|\tilde{\bm{\mu}}_{(2,1)}-\bm{v}_{12}\|+\|\tilde{\bm{\mu}}_{(2,2)}-\bm{v}_{12}\|},\\[3mm]
\varphi_{(1,1)}({\bm v}_{21})-\varphi_{(2,1)}({\bm v}_{21})
&=&\frac{(\epsilon_{21}+\lambda_{21}\cos\Phi)(\lambda_{22}-\epsilon_{22})+(\lambda_{21}^{2}-\epsilon_{21}^{2})}
{\|\tilde{\bm{\mu}}_{(2,1)}-\bm{v}_{21}\|+\|\tilde{\bm{\mu}}_{(1,1)}-\bm{v}_{21}\|-\epsilon_{21}},\\[1.5mm]
\varphi_{(1,2)}({\bm v}_{21})-\varphi_{(2,2)}({\bm v}_{21})
&=&\frac{(\epsilon_{21}-\lambda_{21}\cos\Phi)(\lambda_{22}+\epsilon_{22})+(\lambda_{21}^{2}-\epsilon_{21}^{2})}
{\|\tilde{\bm{\mu}}_{(2,2)}-\bm{v}_{21}\|+\|\tilde{\bm{\mu}}_{(1,2)}-\bm{v}_{21}\|-\epsilon_{21}},\\[3mm]
\varphi_{(1,1)}({\bm v}_{22})-\varphi_{(1,2)}({\bm v}_{22})
&=&\frac{(\epsilon_{22}+\lambda_{22}\cos\Phi)(\lambda_{21}-\epsilon_{21})+(\lambda_{22}^{2}-\epsilon_{22}^{2})}
{\|\tilde{\bm{\mu}}_{(1,2)}-\bm{v}_{22}\|+\|\tilde{\bm{\mu}}_{(1,1)}-\bm{v}_{22}\|-\epsilon_{22}},\\[1.5mm]
\varphi_{(2,1)}({\bm v}_{22})-\varphi_{(2,2)}({\bm v}_{22})
&=&\frac{(\epsilon_{22}-\lambda_{22}\cos\Phi)(\lambda_{21}+\epsilon_{21})+(\lambda_{22}^{2}-\epsilon_{22}^{2})}
{\|\tilde{\bm{\mu}}_{(2,2)}-\bm{v}_{22}\|+\|\tilde{\bm{\mu}}_{(2,1)}-\bm{v}_{22}\|-\epsilon_{22}}.
\end{array}
\end{equation}
Therefore,
\begin{subequations}
\begin{eqnarray}
\mathcal{Z}_{\{(1,1),(1,2)\}}\neq\varnothing &\Leftrightarrow& 
\mbox{$\frac{\epsilon_{21}+\lambda_{21}\cos\Phi}{\lambda_{21}+\epsilon_{21}}\geqslant\frac{\lambda_{21}-\epsilon_{21}}{\lambda_{22}-\epsilon_{22}},\quad
\frac{\epsilon_{21}-\lambda_{21}\cos\Phi}{\lambda_{21}-\epsilon_{21}}\geqslant\frac{\lambda_{21}+\epsilon_{21}}{\lambda_{22}+\epsilon_{22}},$}\\[1mm]
\mathcal{Z}_{\{(1,1),(2,1)\}}\neq\varnothing &\Leftrightarrow& 
\mbox{$\frac{\epsilon_{22}+\lambda_{22}\cos\Phi}{\lambda_{22}+\epsilon_{22}}\geqslant\frac{\lambda_{22}-\epsilon_{22}}{\lambda_{21}-\epsilon_{21}},\quad 
\frac{\epsilon_{22}-\lambda_{22}\cos\Phi}{\lambda_{22}-\epsilon_{22}}\geqslant\frac{\lambda_{22}+\epsilon_{22}}{\lambda_{21}+\epsilon_{21}}$},\\[1mm]
\mathcal{Z}_{\{(2,1),(2,2)\}}\neq\varnothing &\Leftrightarrow& 
\mbox{$\frac{\epsilon_{21}+\lambda_{21}\cos\Phi}{\lambda_{21}+\epsilon_{21}}\leqslant-\frac{\lambda_{21}-\epsilon_{21}}{\lambda_{22}-\epsilon_{22}},\,
\frac{\epsilon_{21}-\lambda_{21}\cos\Phi}{\lambda_{21}-\epsilon_{21}}\leqslant-\frac{\lambda_{21}+\epsilon_{21}}{\lambda_{22}+\epsilon_{22}},$}\label{eq:z2122}\\[1mm]
\mathcal{Z}_{\{(1,2),(2,2)\}}\neq\varnothing &\Leftrightarrow&
\mbox{$\frac{\epsilon_{22}+\lambda_{22}\cos\Phi}{\lambda_{22}+\epsilon_{22}}\leqslant-\frac{\lambda_{22}-\epsilon_{22}}{\lambda_{21}-\epsilon_{21}},\,
\frac{\epsilon_{22}-\lambda_{22}\cos\Phi}{\lambda_{22}-\epsilon_{22}}\leqslant-\frac{\lambda_{22}+\epsilon_{22}}{\lambda_{21}+\epsilon_{21}}.$}\label{eq:z1222}
\end{eqnarray}
\end{subequations}
Moreover, at least one of the two left-hand sides of Inequalities \eqref{eq:z2122} is positive, but the two right-hand sides are both negative; thus, $\mathcal{Z}_{\{(2,1),(2,2)\}}$ is always empty. For similar reasons, $\mathcal{Z}_{\{(1,2),(2,2)\}}$ is always empty.
\end{proof}

\subsection*{Proof of Theorem~\ref{thm:prfthr}}
\renewcommand*{\proofname}{Proof of {\rm (a)}}
\begin{proof}
Suppose $m=n_{1}=n_{2}=2$ and $p_{\rm guess}^{\rm post}<p_{\rm guess}^{\rm prior}$.
Then, we can easily verify that
\begin{equation}\label{eq:cvp}
\begin{array}{lclcl}
\mathcal{X}_{\{(1,1),(2,2)\}}\neq\varnothing&\Rightarrow&
\mathcal{X}_{\{(1,1),(2,2)\}}=\{{\bm v}_{+}\},\ \ \varphi_{(1,1)}(\bm{v}_{+})+\varphi_{(2,2)}(\bm{v}_{+})=1+\gamma_{+},
\\[1mm]
\mathcal{X}_{\{(1,2),(2,1)\}}\neq\varnothing&\Rightarrow&
\mathcal{X}_{\{(1,2),(2,1)\}}=\{{\bm v}_{-}\},\ \ \varphi_{(1,2)}(\bm{v}_{-})+\varphi_{(2,1)}(\bm{v}_{-})=1+\gamma_{-},
\end{array}
\end{equation}
where
\begin{equation}
\begin{array}{rclcrcl}
{\bm v}_{+}&=&\Big(\frac{\gamma_{+}+\epsilon_{21}+\epsilon_{22}}{2\gamma_{+}}\Big)\tilde{\bm{\mu}}_{(1,1)}+\Big(\frac{\gamma_{+}-\epsilon_{21}-\epsilon_{22}}{2\gamma_{+}}\Big)\tilde{\bm{\mu}}_{(2,2)},\\[1.5mm]
{\bm v}_{-}&=&\Big(\frac{\gamma_{-}+\epsilon_{21}-\epsilon_{22}}{2\gamma_{-}}\Big)\tilde{\bm{\mu}}_{(1,2)}+\Big(\frac{\gamma_{-}-\epsilon_{21}+\epsilon_{22}}{2\gamma_{-}}\Big)\tilde{\bm{\mu}}_{(2,1)}.
\end{array}
\end{equation}
Since $\mathcal{X}_{\mathsf{S}}\cap\mathcal{Y}_{\mathsf{S}}=\mathcal{Z}_{\mathsf{S}}$ for all $\mathsf{S}\subseteq\Omega$, Eq.~\eqref{eq:cvp} implies
\begin{equation}\label{eq:zvy}
\begin{array}{lclcl}
\mathcal{Z}_{\{(1,1),(2,2)\}}\neq\varnothing&\Rightarrow&
{\bm v}_{+}\in\mathcal{Y}_{\{(1,1),(2,2)\}},\\[1mm]
\mathcal{Z}_{\{(1,2),(2,1)\}}\neq\varnothing&\Rightarrow&
{\bm v}_{-}\in\mathcal{Y}_{\{(1,2),(2,1)\}},
\end{array}
\end{equation}
as well as the validity of Eq.~\eqref{eq:mspth}.\\
\indent From definition of $\mathcal{Y}_{\mathsf{S}}$, we have
\begin{equation}\label{eq:vyv}
\begin{array}{rcl}
\bm{v}_{+}\in\mathcal{Y}_{\{(1,1),(2,2)\}} &\Leftrightarrow 
&\varphi_{(1,1)}({\bm v}_{+})\geqslant\varphi_{(1,2)}({\bm v}_{+}),\ 
\varphi_{(1,1)}({\bm v}_{+})\geqslant\varphi_{(2,1)}({\bm v}_{+}),\\[1.5mm]
\bm{v}_{-}\in\mathcal{Y}_{\{(1,2),(2,1)\}} &\Leftrightarrow 
&\varphi_{(1,1)}({\bm v}_{-})\leqslant\varphi_{(1,2)}({\bm v}_{-}),\ 
\varphi_{(2,2)}({\bm v}_{-})\leqslant\varphi_{(1,2)}({\bm v}_{-}).
\end{array}
\end{equation}
It is straightforward to show that
\begin{equation}\label{eq:vvrab}
\begin{array}{rcl}
\varphi_{(1,1)}({\bm v}_{+})-\varphi_{(1,2)}({\bm v}_{+})
&=&\frac{\gamma_{+}\alpha-\beta_{+}}
{\gamma_{+}(\epsilon_{21}+\|\tilde{\bm{\mu}}_{(1,1)}-\bm{v}_{+}\|+\|\tilde{\bm{\mu}}_{(1,2)}-\bm{v}_{+}\|)},\\[1.5mm]
\varphi_{(1,1)}({\bm v}_{+})-\varphi_{(2,1)}({\bm v}_{+})
&=&\frac{\gamma_{+}\alpha+\beta_{+}}
{\gamma_{+}(\epsilon_{21}+\|\tilde{\bm{\mu}}_{(1,1)}-\bm{v}_{+}\|+\|\tilde{\bm{\mu}}_{(2,1)}-\bm{v}_{+}\|)},\\[1.5mm]
\varphi_{(1,1)}({\bm v}_{-})-\varphi_{(1,2)}({\bm v}_{-})
&=&\frac{\gamma_{-}\alpha+\beta_{-}}
{\gamma_{-}(\|\tilde{\bm{\mu}}_{(1,2)}-\bm{v}_{-}\|+\|\tilde{\bm{\mu}}_{(1,1)}-\bm{v}_{-}\|-\epsilon_{22})},\\[1.5mm]
\varphi_{(1,1)}({\bm v}_{-})-\varphi_{(2,1)}({\bm v}_{-})
&=&\frac{\gamma_{-}\alpha-\beta_{-}}
{\gamma_{-}(\epsilon_{21}+\|\tilde{\bm{\mu}}_{(1,2)}-\bm{v}_{-}\|+\|\tilde{\bm{\mu}}_{(2,2)}-\bm{v}_{-}\|)}.
\end{array}
\end{equation}
By applying Eq.~\eqref{eq:vvrab} to \eqref{eq:vyv} and \eqref{eq:vyv} to \eqref{eq:zvy},
we can show that \eqref{eq:nscthrpm}($\Rightarrow$) is true.\\
\indent Now, we prove that \eqref{eq:nscthrpm}($\Leftarrow$) is true.
Because non-negativity and non-positivity of $\alpha$ imply $\gamma_{+}^{2}\geqslant(\epsilon_{21}+\epsilon_{22})^{2}$ and $\gamma_{-}^{2}\geqslant(\epsilon_{21}-\epsilon_{22})^{2}$, respectively,
\begin{equation}\label{eq:axvpm}
\begin{array}{lclcl}
\alpha\geqslant0&\Rightarrow&\mathcal{X}_{\{(1,1),(2,2)\}}=\{\bm{v}_{+}\},\\[1.5mm]
\alpha\leqslant0&\Rightarrow&\mathcal{X}_{\{(1,2),(2,1)\}}=\{\bm{v}_{-}\}.
\end{array}
\end{equation}
Thus, both $\mathcal{X}_{\{(1,1),(2,2)\}}$ and $\mathcal{Y}_{\{(1,1),(2,2)\}}$ have $\bm{v}_{+}$ by Eq.~\eqref{eq:vvrab} and \eqref{eq:axvpm} if $\alpha\geqslant|\beta_{+}|/\gamma_{+}$.
Similarly, both $\mathcal{X}_{\{(1,2),(2,1)\}}$ and $\mathcal{Y}_{\{(1,2),(2,1)\}}$ have $\bm{v}_{-}$
if $\alpha\leqslant-|\beta_{-}|/\gamma_{-}$. 
Therefore, \eqref{eq:nscthrpm}($\Leftarrow$) is also proved.
\end{proof}
\renewcommand*{\proofname}{Proof}
\renewcommand*{\proofname}{Proof of {\rm(b)}}
\begin{proof}
Assume $m=n_{1}=n_{2}=2$ and $p_{\rm guess}^{\rm post}<p_{\rm guess}^{\rm prior}$.
We can see from Lemma~\ref{lem:xnep} and Eq.~\eqref{eq:xvvv} that
\begin{equation}\label{eq:xvpv}
\begin{array}{ccc}
\mathcal{X}_{\{(1,1),(1,2),(2,1)\}}\neq\varnothing&\Rightarrow&
\varphi_{(1,1)}({\bm v}_{11})<\varphi_{(2,1)}({\bm v}_{11}),\
\varphi_{(1,1)}({\bm v}_{12})<\varphi_{(1,2)}({\bm v}_{12}),
\\[1.5mm]
\mathcal{X}_{\{(1,1),(1,2),(2,2)\}}\neq\varnothing&\Rightarrow&
\varphi_{(1,2)}({\bm v}_{11})<\varphi_{(2,2)}({\bm v}_{11}),\ 
\varphi_{(1,2)}({\bm v}_{22})<\varphi_{(1,1)}({\bm v}_{22}),
\\[1.5mm]
\mathcal{X}_{\{(1,1),(2,1),(2,2)\}}\neq\varnothing&\Rightarrow&
\varphi_{(2,1)}({\bm v}_{12})<\varphi_{(2,2)}({\bm v}_{12}),\ 
\varphi_{(2,1)}({\bm v}_{21})<\varphi_{(1,1)}({\bm v}_{21}),
\\[1.5mm]
\mathcal{X}_{\{(1,2),(2,1),(2,2)\}}\neq\varnothing&\Rightarrow&
\varphi_{(2,2)}({\bm v}_{21})<\varphi_{(1,2)}({\bm v}_{21}),\
\varphi_{(2,2)}({\bm v}_{22})<\varphi_{(2,1)}({\bm v}_{22}).
\end{array}
\end{equation}
By applying Eq.~\eqref{eq:vvvvc2} to \eqref{eq:xvpv},
we have
\begin{equation}\label{eq:ww}
\begin{array}{ccl}
\mathcal{X}_{\{(1,1),(1,2),(2,1)\}}\neq\varnothing&\Rightarrow&
\frac{\epsilon_{21}+\lambda_{21}\cos\Phi}{\lambda_{21}+\epsilon_{21}}<\frac{\lambda_{21}-\epsilon_{21}}{\lambda_{22}-\epsilon_{22}},\quad
\frac{\epsilon_{22}+\lambda_{22}\cos\Phi}{\lambda_{22}+\epsilon_{22}}<\frac{\lambda_{22}-\epsilon_{22}}{\lambda_{21}-\epsilon_{21}},
\\[1mm]
\mathcal{X}_{\{(1,1),(1,2),(2,2)\}}\neq\varnothing&\Rightarrow&
\frac{\epsilon_{21}-\lambda_{21}\cos\Phi}{\lambda_{21}-\epsilon_{21}}<\frac{\lambda_{21}+\epsilon_{21}}{\lambda_{22}+\epsilon_{22}},\quad
\frac{\epsilon_{22}+\lambda_{22}\cos\Phi}{\lambda_{22}+\epsilon_{22}}>-\frac{\lambda_{22}-\epsilon_{22}}{\lambda_{21}-\epsilon_{21}},
\\[1mm]
\mathcal{X}_{\{(1,1),(2,1),(2,2)\}}\neq\varnothing&\Rightarrow&
\frac{\epsilon_{22}-\lambda_{22}\cos\Phi}{\lambda_{22}-\epsilon_{22}}<\frac{\lambda_{22}+\epsilon_{22}}{\lambda_{21}+\epsilon_{21}},\quad
\frac{\epsilon_{21}+\lambda_{21}\cos\Phi}{\lambda_{21}+\epsilon_{21}}>-\frac{\lambda_{21}-\epsilon_{21}}{\lambda_{22}-\epsilon_{22}},
\\[1mm]
\mathcal{X}_{\{(1,2),(2,1),(2,2)\}}\neq\varnothing&\Rightarrow&
\frac{\epsilon_{21}-\lambda_{21}\cos\Phi}{\lambda_{21}-\epsilon_{21}}>-\frac{\lambda_{21}+\epsilon_{21}}{\lambda_{22}+\epsilon_{22}},\
\frac{\epsilon_{22}-\lambda_{22}\cos\Phi}{\lambda_{22}-\epsilon_{22}}>-\frac{\lambda_{22}+\epsilon_{22}}{\lambda_{21}+\epsilon_{21}}.
\end{array}
\end{equation}
\indent Since the inequalities given in \eqref{eq:ww} can also be expressed as
\begin{equation}
\begin{array}{lcl}
\frac{\epsilon_{22}\mp\lambda_{22}\cos\Phi}{\lambda_{22}\mp\epsilon_{22}}<\frac{\lambda_{22}\pm\epsilon_{22}}{\lambda_{21}\pm\epsilon_{21}}
&\Leftrightarrow&\Theta_{\pm}>\Xi_{\pm},
\\[1mm]
\frac{\epsilon_{21}\mp\lambda_{21}\cos\Phi}{\lambda_{21}\mp\epsilon_{21}}<\frac{\lambda_{21}\pm\epsilon_{21}}{\lambda_{22}\pm\epsilon_{22}}
&\Leftrightarrow&\mp\Theta_{\pm}<\Phi\mp\Xi_{\pm}<\pi\mp\pi\pm\Theta_{\pm},
\\[1mm]
\frac{\epsilon_{22}\pm\lambda_{22}\cos\Phi}{\lambda_{22}\pm\epsilon_{22}}>-\frac{\lambda_{22}\mp\epsilon_{22}}{\lambda_{21}\mp\epsilon_{21}}
&\Leftrightarrow&\Theta_{\pm}<\pi-\Xi_{\pm},
\\[1mm]
\frac{\epsilon_{21}\pm\lambda_{21}\cos\Phi}{\lambda_{21}\pm\epsilon_{21}}>-\frac{\lambda_{21}\mp\epsilon_{21}}{\lambda_{22}\mp\epsilon_{22}}
&\Leftrightarrow&\mp\pi\pm\Theta_{\pm}<\Phi\mp\Xi_{\pm}<\pi\mp\Theta_{\pm},
\end{array}
\end{equation}
a lengthy calculation show that
\begin{equation}\label{eq:vxuv}
\begin{array}{ccccc}
\bm{v}\in\mathcal{X}_{\{(1,1),(1,2),(2,1)\}}&\Rightarrow&
\left\{
\begin{array}{ccc}
\angle\bm{v}\tilde{\bm{\mu}}_{(1,1)}\tilde{\bm{\mu}}_{(2,1)}&=&\Theta_{-}-\Xi_{-},\\[1.5mm]
\|\tilde{\bm{\mu}}_{(1,1)}-\bm{v}\|
&=&\frac{\lambda_{21}^{2}-\epsilon_{21}^{2}}{2[\lambda_{21}\cos(\Theta_{-}-\Xi_{-})+\epsilon_{21}]}
,\\[1.5mm]
\varphi_{(1,1)}({\bm v})-\varphi_{(2,2)}({\bm v})&=&\frac{-2\alpha}{\|\tilde{\mu}_{(1,1)}-\bm{v}\|+\|\tilde{\mu}_{(2,2)}-\bm{v}\|+\epsilon_{21}+\epsilon_{22}},
\end{array}
\right.\\[8mm]
\bm{v}\in\mathcal{X}_{\{(1,1),(1,2),(2,2)\}}&\Rightarrow&
\left\{
\begin{array}{ccc}
\angle\bm{v}\tilde{\bm{\mu}}_{(1,2)}\tilde{\bm{\mu}}_{(2,2)}&=&\pi-\Theta_{+}-\Xi_{+},\\[1.5mm]
\|\tilde{\bm{\mu}}_{(1,2)}-\bm{v}\|
&=&\frac{\lambda_{21}^{2}-\epsilon_{21}^{2}}{2[\lambda_{21}\cos(\pi-\Theta_{+}-\Xi_{+})+\epsilon_{21}]}
,\\[1.5mm]
\varphi_{(1,2)}({\bm v})-\varphi_{(2,1)}({\bm v})&=&\frac{+2\alpha}{\|\tilde{\bm{\mu}}_{(1,2)}-\bm{v}\|+\|\tilde{\bm{\mu}}_{(2,1)}-\bm{v}\|+\epsilon_{21}-\epsilon_{22}},
\end{array}
\right.\\[8mm]
\bm{v}\in\mathcal{X}_{\{(1,1),(2,1),(2,2)\}}&\Rightarrow&
\left\{
\begin{array}{ccc}
\angle\bm{v}\tilde{\bm{\mu}}_{(2,1)}\tilde{\bm{\mu}}_{(1,1)}&=&\Theta_{+}-\Xi_{+},\\[1.5mm]
\|\tilde{\bm{\mu}}_{(2,1)}-\bm{v}\|
&=&\frac{\lambda_{21}^{2}-\epsilon_{21}^{2}}{2[\lambda_{21}\cos(\Theta_{+}-\Xi_{+})-\epsilon_{21}]}
,\\[1.5mm]
\varphi_{(2,1)}({\bm v})-\varphi_{(1,2)}({\bm v})&=&\frac{+2\alpha}{\|\tilde{\bm{\mu}}_{(2,1)}-\bm{v}\|+\|\tilde{\bm{\mu}}_{(1,2)}-\bm{v}\|-\epsilon_{21}+\epsilon_{22}},
\end{array}
\right.\\[8mm]
\bm{v}\in\mathcal{X}_{\{(1,2),(2,1),(2,2)\}}&\Rightarrow&
\left\{
\begin{array}{ccc}
\angle\bm{v}\tilde{\bm{\mu}}_{(2,2)}\tilde{\bm{\mu}}_{(1,2)}&=&\pi-\Theta_{-}-\Xi_{-},\\[1.5mm]
\|\tilde{\bm{\mu}}_{(2,2)}-\bm{v}\|
&=&\frac{\lambda_{21}^{2}-\epsilon_{21}^{2}}{2[\lambda_{21}\cos(\pi-\Theta_{-}-\Xi_{-})-\epsilon_{21}]}
,\\[1.5mm]
\varphi_{(2,2)}({\bm v})-\varphi_{(1,1)}({\bm v})&=&
\frac{-2\alpha}{\|\tilde{\bm{\mu}}_{(2,2)}-\bm{v}\|+\|\tilde{\bm{\mu}}_{(1,1)}-\bm{v}\|
-\epsilon_{21}-\epsilon_{22}}.
\end{array}
\right.
\end{array}
\end{equation}
Here, $\angle\bm{v}\tilde{\bm{\mu}}_{\bm{\omega}}\tilde{\bm{\mu}}_{\bm{\omega}'}$
is the internal angle between two line segments formed by $\{\bm{v},\tilde{\bm{\mu}}_{\bm{\omega}}\}$ and $\{\tilde{\bm{\mu}}_{\bm{\omega}},\tilde{\bm{\mu}}_{\bm{\omega}'}\}$, respectively.\\
\indent As $\mathcal{X}_{\mathsf{S}}\cap\mathcal{Y}_{\mathsf{S}}=\mathcal{Z}_{\mathsf{S}}$ for all $\mathsf{S}\subseteq\Omega$,
Eq.~\eqref{eq:mspth1} is true from Lemma~\ref{lem:optmea} and Eq.~\eqref{eq:vxuv}. 
Moreover, Eq. \eqref{eq:vxuv} also implies
\begin{equation}\label{eq:zza0}
\begin{array}{ccc}
\mathcal{Z}_{\{(1,1),(1,2),(2,1)\}}\neq\varnothing
&\Leftrightarrow&
\mathcal{X}_{\{(1,1),(1,2),(2,1)\}}\neq\varnothing,\
\alpha\leqslant0,
\\[1mm]
\mathcal{Z}_{\{(1,1),(1,2),(2,2)\}}\neq\varnothing
&\Leftrightarrow&
\mathcal{X}_{\{(1,1),(1,2),(2,2)\}}\neq\varnothing,\ 
\alpha\geqslant0,
\\[1mm]
\mathcal{Z}_{\{(1,1),(2,1),(2,2)\}}\neq\varnothing
&\Leftrightarrow&
\mathcal{X}_{\{(1,1),(2,1),(2,2)\}}\neq\varnothing,\ 
\alpha\geqslant0,
\\[1mm]
\mathcal{Z}_{\{(1,2),(2,1),(2,2)\}}\neq\varnothing
&\Leftrightarrow&
\mathcal{X}_{\{(1,2),(2,1),(2,2)\}}\neq\varnothing,\ 
\alpha\leqslant0,
\end{array}
\end{equation}
Now, \eqref{eq:axvpm} and \eqref{eq:zza0} lead us to
\begin{equation}\label{eq:zzxv}
\begin{array}{ccc}
\mathcal{Z}_{\{(1,1),(1,2),(2,1)\}}\neq\varnothing
\ \mbox{or}\ 
\mathcal{Z}_{\{(1,2),(2,1),(2,2)\}}\neq\varnothing
&\Rightarrow&
\mathcal{X}_{\{(1,2),(2,1)\}}=\{\bm{v}_{-}\},
\\[1mm]
\mathcal{Z}_{\{(1,1),(1,2),(2,2)\}}\neq\varnothing
\ \mbox{or}\ 
\mathcal{Z}_{\{(1,1),(2,1),(2,2)\}}\neq\varnothing
&\Rightarrow&
\mathcal{X}_{\{(1,2),(2,1)\}}=\{\bm{v}_{+}\}.
\end{array}
\end{equation}
Thus, we can show from Eq.~\eqref{eq:vvrab}, \eqref{eq:zzxv}, and Lemma~\ref{lem:xnep} that
\begin{equation}\label{eq:zepa}
\begin{array}{rclcccc}
\mathcal{Z}_{\{(1,1),(1,2),(2,1)\}}\neq\varnothing&\Rightarrow&
\alpha>-\beta_{-}/\gamma_{-},
\\[1mm]
\mathcal{Z}_{\{(1,1),(1,2),(2,2)\}}\neq\varnothing&\Rightarrow&
\alpha<-\beta_{+}/\gamma_{+},\\[1mm]
\mathcal{Z}_{\{(1,1),(2,1),(2,2)\}}\neq\varnothing&\Rightarrow&
\alpha<\beta_{+}/\gamma_{+},
\\[1mm]
\mathcal{Z}_{\{(1,2),(2,1),(2,2)\}}\neq\varnothing&\Rightarrow&
\alpha>\beta_{-}/\gamma_{-}.
\end{array}
\end{equation}
Up to this point, \eqref{eq:nscthrpm1}($\Rightarrow$) 
for $\mathsf{S}\neq\{(1,2),(2,1),(2,2)\}$ is proved by
\eqref{eq:ww}, \eqref{eq:zza0}, and \eqref{eq:zepa}.\\
\indent Since $\beta_{-}/\gamma_{-}<\alpha\leqslant0$ 
is a necessary condition for $\mathcal{Z}_{\{(1,2),(2,1),(2,2)\}}\neq\varnothing$ from \eqref{eq:zza0} and \eqref{eq:zepa}, $\mathcal{Z}_{\{(1,2),(2,1),(2,2)\}}\neq\varnothing$ requires non-positivity of $\alpha$ and $\beta_{-}$. 
However, both $\alpha$ and $\beta_{-}$ cannot be non-positive because
\begin{equation}\label{eq:aebel}
\alpha(\epsilon_{21}+\epsilon_{22})+\beta_{-}
=\epsilon_{21}(\lambda_{22}^{2}-\epsilon_{22}^{2})
+\epsilon_{22}(\lambda_{21}^{2}-\epsilon_{21}^{2})>0.
\end{equation}
Thus, $\mathcal{Z}_{\{(1,2),(2,1),(2,2)\}}$ is always empty.\\
\indent Now, for the proof of \eqref{eq:nscthrpm1}($\Leftarrow$),
let $\mathsf{S}$ be a proper subset of $\Omega$ 
satisfying the necessary condition of \eqref{eq:nscthrpm1}. 
Then, $\mathcal{X}_{\mathsf{S}'}$ is nonempty 
for all $\mathsf{S}'\subsetneq\mathsf{S}$ with $|\mathsf{S}'|=2$ from \eqref{eq:xvvv} and \eqref{eq:axvpm}. Thus, we can see from Lemma~\ref{lem:sufc} and \eqref{eq:zza0} that $\mathcal{Z}_{\mathsf{S}}$ is nonempty if $\mathcal{X}_{\mathsf{S}'}^{\mathsf{S}}$ is empty for all $\mathsf{S}'\subsetneq\mathsf{S}$ with $|\mathsf{S}'|=2$.
Using \eqref{eq:vvvvc2} and \eqref{eq:vvrab}, we can easily show that
$\mathcal{X}_{\mathsf{S}'}^{\mathsf{S}}=\varnothing$ for all $\mathsf{S}'\subsetneq\mathsf{S}$ with $|\mathsf{S}'|=2$. Therefore, $\mathcal{Z}_{\mathsf{S}}\neq\varnothing$. That is, \eqref{eq:nscthrpm1}($\Leftarrow$) is proved.
\end{proof}
\renewcommand*{\proofname}{Proof}
\renewcommand*{\proofname}{Proof of {\rm (c)}}
\begin{proof}
\indent Let us assume $\mathcal{Z}\subseteq\mathcal{P}_{\Omega}$.
Since a null MEPI measurement exists from Corollary~\ref{cor:null},
$\mathcal{Z}_{\mathsf{S}}\neq\varnothing$ for some $\mathsf{S}$ in \eqref{eq:zinp} with $\mathsf{S}\neq\Omega$.
When $\bm{v}$ is the single element  in $\mathcal{Z}_{\mathsf{S}}$,
we can see from the definition of $\mathcal{Z}_{\Omega}$ that $\mathcal{Z}_{\Omega}\neq\varnothing$ if and only if $\varphi_{\bm{\omega}}(\bm{v})=\varphi_{\bm{\omega}'}(\bm{v})$ for all $(\bm{\omega},\bm{\omega}')\in\mathsf{S}\times(\Omega-\mathsf{S})$. Thus,
when $|\mathsf{S}|=3$, we can show from \eqref{eq:vxuv} that $\mathcal{Z}_{\Omega}\neq\varnothing$ is equivalent to $\alpha=0$.\\
\indent When $|\mathsf{S}|=2$, we can verify from Eqs.~\eqref{eq:cvp} and \eqref{eq:vvrab} that $\mathcal{Z}_{\Omega}\neq\varnothing$ is equivalent to $\alpha=0$. Note that $\alpha=0$ means $\beta_{+}=0$
when $\mathsf{S}=\{(1,1),(2,2)\}$
because $\alpha\geqslant|\beta_{+}|/\gamma_{+}$ is a necessary condition for $\mathcal{Z}_{\{(1,1),(2,2)\}}\neq\varnothing$.
We also note that $\alpha=0$ implies $\beta_{-}=0$
when $\mathsf{S}=\{(1,2),(2,1)\}$
because $\alpha\leqslant-|\beta_{-}|/\gamma_{-}$ is necessary for $\mathcal{Z}_{\{(1,2),(2,1)\}}\neq\varnothing$, therefore $\alpha\neq0$ and $\mathcal{Z}_{\Omega}=\varnothing$
when $\mathsf{S}=\{(1,2),(2,1)\}$ since both $\alpha$ and $\beta_{-}$ cannot be zero by Inequality \eqref{eq:aebel}.
Thus, $\alpha=0$ is a necessary and sufficient condition for $\mathcal{Z}_{\Omega}\neq\varnothing$ when $\mathcal{Z}\subseteq\mathcal{P}_{\Omega}$. 
This completes our proof of Theorem~\ref{thm:prfthr}(c).
\end{proof}
\renewcommand*{\proofname}{Proof}


\end{document}